\documentclass{article}

\usepackage{amsfonts}
\usepackage{amsmath}
\usepackage{amsthm}
\usepackage{arxiv}
\usepackage{booktabs}
\usepackage{doi}
\usepackage[T1]{fontenc}
\usepackage{graphicx}
\usepackage{hyperref}
\usepackage[utf8]{inputenc}
\usepackage{mathrsfs}
\usepackage{microtype}
\usepackage{nicefrac}
\usepackage{url}
\usepackage{xcolor}

\theoremstyle{plain}
\newtheorem{proposition}{Proposition}
\newtheorem{lemma}{Lemma}
\newtheorem{theorem}{Theorem}
\newtheorem{corollary}{Corollary}
\newtheorem{conjecture}{Conjecture}

\theoremstyle{definition}
\newtheorem{definition}{Definition}

\theoremstyle{remark}

\title{General Relativistic Shock Waves that Exhibit an Accelerated Expansion}

\author{ \href{https://orcid.org/0000-0001-9255-6281}{\includegraphics[scale=0.06]{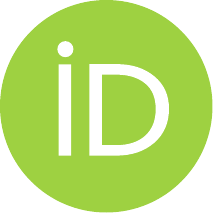}\hspace{1mm}Christopher Alexander}\\
Department of Mathematics\\
Rutgers University\\
Piscataway, NJ 08854\\
\texttt{christopher.alexander@rutgers.edu}\\
}

\renewcommand{\headeright}{}
\renewcommand{\undertitle}{}
\renewcommand{\shorttitle}{General Relativistic Shock Waves that Exhibit an Accelerated Expansion}

\hypersetup{
	pdftitle={General Relativistic Shock Waves that Exhibit an Accelerated Expansion},
	pdfsubject={gr-qc, math-ph},
	pdfauthor={Christopher Alexander},
	pdfkeywords={General Relativity, Self Similar, Shock Wave, Cosmology, Dark Energy},
}

\begin{document}

\maketitle

\begin{abstract}

This paper concerns the construction and analysis of a new family of exact general relativistic shock waves. The construction resolves the long-standing open problem of determining the expanding waves created behind a shock-wave explosion within a static isothermal sphere, akin to stellar ignition. In particular, these shock-wave spacetimes are formed through matching a family of self-similar asymptotically Friedmann spacetimes to the family of Tolman--Oppenheimer--Volkoff spacetimes. The matching is accomplished in Schwarzschild coordinates where the shock waves appear one derivative less regular than they actually are. Separately, both families contain singularities, but as matched shock-wave spacetimes, they are singularity free. It was previously unknown whether the matching of the two families could be achieved within a region where both families are nonsingular. Indeed, for a pure radiation equation of state, the matching occurs near the sonic point of the interior expanding wave and this makes the analysis quite delicate, both numerically and formally. It is for this reason the construction is accompanied by a rigorous existence proof in the pure radiation case, which forms the main result of this paper. The analysis extends to consider the matching of any self-similar expanding wave with a Tolman--Oppenheimer--Volkoff spacetime, with potentially distinct equations of state each side of the shock. The asymptotically Friedmann shock-waves exhibit an accelerated expansion, analogous to the acceleration modelled by the cosmological constant in the Standard Model of Cosmology. However, unlike in the Standard Model, these shock-wave spacetimes solve the Einstein field equations in the absence of a cosmological constant, opening up the question of whether a purely mathematical mechanism could account for the accelerated expansion observed today, rather than dark energy.

\end{abstract}

\keywords{General Relativity \and Self Similar \and Shock Wave \and Cosmology \and Dark Energy}

This material is based upon work supported by the National Science Foundation under Grant No. 1809311.

\vfill

\pagebreak

\tableofcontents

\vfill

\pagebreak

\section{Introduction}\label{Section1}

Consider a static isothermal sphere with an inverse square density profile held in gravitational equilibrium by its own pressure. These spherically symmetric solutions of the Einstein--Euler equations form a subset of the Tolman--Oppenheimer--Volkoff (TOV) spacetimes, some of which are considered as idealised models for stars. The spacetimes with an inverse square density profile are physically unstable due to the singularity in pressure and density at the radial centre. In this paper, the singular core is removed and replaced by an expanding wave, separated from the rest of the static spacetime by a spherical shock surface. The resulting spacetime is the general relativistic analogue of a shock-wave explosion within a static isothermal sphere. In particular, what is determined is the physically relevant family of interior expanding waves that can be matched to an exterior static sphere to form a global general relativistic shock wave. Most importantly, these waves are determined in the case of pure radiation, the state of matter in the Radiation Dominated Epoch of the Early Universe.

These shock waves conserve mass-energy and momentum across the shock surface and produce no delta function sources in Schwarzschild coordinates. Moreover, there exists a local coordinate transformation for which the solution has optimal metric regularity, which is $C^{1,1}$. Such properties make these shock waves true weak solutions of the Einstein--Euler equations, with the expanding waves given as exact, but not explicit, solutions of a system of ordinary differential equations (ODE). For a more in-depth consideration of optimal metric regularity, see \cite{RT2020A,RT2020B,RT2020C,RT2020D} and the subsequent papers of Reintjes and Temple.

The physically relevant family of expanding waves is a subset of a larger family of spherically symmetric self-similar spacetimes first considered by Cahill and Taub in 1971 \cite{CT1971}, but the problem of determining whether members of this family are the actual expanding waves of outgoing shock waves has remained open. This paper resolves this open problem for shock waves expanding outward into a static singular isothermal sphere such that the global solution regularises the singularity at the radial centre. The physically relevant family consists of spacetimes that are asymptotically Friedmann approaching the radial origin at some fixed time, or approaching the distant future at some fixed radius. This family contains the flat Friedmann--Lemaître--Robertson--Walker (FLRW) spacetime, representing the current Standard Model of Cosmology, with all other members corresponding to self-similar perturbations of this spacetime. This family was first identified by Carr and Yahil \cite{CY1990}, and later independently identified by Smoller and Temple \cite{TS2009}. Carr and Yahil interpret these perturbations as density perturbations, whereas Smoller and Temple interpret them as perturbations to the magnitude of accelerated expansion of the spacetime.

Smoller and Temple demonstrate in \cite{ST2012A} that the parameter corresponding to the magnitude of accelerated expansion, referred to as the \emph{acceleration parameter}, mimics the acceleration of a flat FLRW spacetime with positive cosmological constant to leading order. Given that the inclusion of a shock wave into a cosmological model is a simple and natural way to bound the total mass of a big bang, it is conjectured that the anomalous acceleration incurred in the expansion wave created behind a big bang shock wave is responsible for the observed accelerated expansion, rather than an acceleration induced by some form of dark energy. A dual-state model in the form of a general relativistic shock wave is essential for this conjectured mechanism of accelerated expansion, noting that the limited parameter freedom in spherically symmetric self-similar perfect fluid spacetimes determines the acceleration parameter upon fixing the equation of state each side of the shock. It would be remarkable if the accelerations specified by a shock wave model matched the accelerations currently observed. One caveat is that such a model requires our galaxy to be close to the centre of expansion, as this is the only region where the spacetime is almost flat. This breaks the Copernican principle in a spatial sense, by placing our galaxy in a special place in the observable universe. However, it does not place our galaxy in a special place in time, as is the case for the Standard Model of Cosmology, which assumes our galaxy exists during a time when the mass-energy contribution of dark energy is on the same order of magnitude as matter.

Such a \emph{Big Wave} or \emph{Big Shock} cosmological model places the present observable Universe within the expanding shock surface, which, if subluminal, will eventually come into view \cite{ST2012B}. This paper works towards such a model by constructing an analogous shock wave for which the shock surface lies within the Hubble radius, that is, the shock surface would already be observable if taken as a cosmological model. The motivation for considering shock waves in the Radiation Dominated Epoch is that the accelerated expansion exhibited in this epoch influences the accelerated expansion found in the Matter Dominated Epoch that we inhabit today. This connection is considered by Alexander, Temple and Vogler in \cite{ATV2026}, with an interesting consequence of this paper being the realisation that the flat FLRW spacetime in the Matter Dominated Epoch is unstable to spherical perturbations. This opens up the question of what alternative spacetimes in the Matter Dominated Epoch we are more likely to find ourselves in, and in particular, what spacetimes in the Radiation Dominated Epoch give rise to them.

The construction of a shock wave in the Radiation Dominated Epoch with an asymptotically Friedmann interior and a shock surface beyond the Hubble radius, along with formalising the details of the transition into the Matter Dominated Epoch, are topics of current research for Alexander, Temple and Vogler. It is useful to note that Smoller and Temple have already identified a method for matching an FLRW and a static spacetime beyond the Hubble radius in \cite{ST2003}, with the details provided in \cite{ST2004}. This method involves placing the whole Universe, including the shock surface, within the Schwarzschild radius of a time reversed black hole. However, Smoller and Temple did not fully resolve the expansion wave behind the shock for pure radiation equations of state each side of the shock surface.

In this paper, solutions of the Einstein--Euler equations are assumed to be spherically symmetric and self-similar in the variable $\xi=\frac{r}{t}$, that is, self-similarity of the first kind. These two assumptions reduce the Einstein--Euler equations, a system of nonlinear partial differential equations, to a system of nonlinear ordinary differential equations in the single variable $\xi$. It is under these assumptions that Cahill and Taub in \cite{CT1971} have established criteria for the uniqueness of solutions. The flat FLRW and TOV spacetimes can be written as explicit solutions to these equations when coupled with barotropic equations of state. For self-similar solutions of the first kind, barotropic equations of state are restricted to the form $p=\sigma\rho$ for some constant $\sigma$, which is referred to as the \emph{equation of state parameter}. It is also shown in \cite{CT1971} that the TOV spacetimes form the unique family of static perfect fluid spacetimes that are spherically symmetric and self-similar of the first kind. This family forms the exterior of all general relativistic shock waves considered in this paper.

Cahill and Taub were also the first to construct an explicit general relativistic shock wave by matching the explicitly known pure radiation ($\sigma=\frac{1}{3}$) flat FLRW spacetime to an explicitly known TOV spacetime. In addition, they claimed the existence of a two-parameter family of pure radiation self-similar spacetimes that could be matched to a TOV spacetime to form a shock wave in a subsequent paper that was not published and possibly never completed. This paper resolves this open problem for the one-parameter subfamily of pure radiation asymptotically Friedmann spacetimes.

Subsequent to Cahill and Taub's pioneering paper, many substantial advances in the field of general relativistic shock waves came from Smoller and Temple throughout the 1990s. In \cite{ST1994}, a number of theorems concerning the regularity of spherically symmetric general relativistic shock waves are given, and in \cite{ST1995}, a criteria for determining the Lax characteristic conditions of these shock waves is introduced. Furthermore, Smoller and Temple in \cite{ST1995} generalise Cahill and Taub's shock wave to a one-parameter family of shock waves, with the free parameter corresponding to either the interior equation of state parameter or the exterior equation of state parameter, but not both. This is an important detail as it implies an additional parameter is needed to construct a general relativistic shock wave in the Radiation Dominated Epoch, that is, a shock wave with a pure radiation equation of state on both the interior and exterior side of the shock surface. The statement of these results are provided in the next section.

This additional parameter is identified by Smoller and Temple in \cite{ST2012A}, which derives the family of asymptotically Friedmann spacetimes in self-similar Schwarzschild coordinates local to the radial centre. This family is a one-parameter or two-parameter family depending on whether the equation of state parameter is included in the count. It is common in the literature to exclude this parameter from the count but this paper does not follow this convention and instead makes clear whether the equation of state parameter is included in each instance. This additional parameter is the aforementioned acceleration parameter and it is in \cite{ST2012A} that the comparison between the acceleration parameter and the cosmological constant is formalised.

Unbeknown to Smoller and Temple at the time, Carr and Yahil already identified the family of asymptotically Friedmann spacetimes in self-similar comoving coordinates, using the same convention employed by Cahill and Taub. However, unlike the exact local form derived by Smoller and Temple, Carr and Yahil instead derived an asymptotic local form. The complete classification of spherically symmetric self-similar solutions of the first kind was later completed by Carr and Coley in 2000 \cite{CC2000}. In addition to determining the number of free parameters present in each family of spacetimes, they also provide a detailed discussion of their physical relevance.

There is a fair amount of preliminary theory that needs to be introduced prior to the construction and analysis of the new family of general relativistic shock waves, and this is the purpose of Section \ref{Section2}. The system of ODE introduced by Cahill and Taub, and used by Carr and others, are derived using self-similar comoving coordinates, whereas the ODE introduced by Smoller and Temple \cite{ST2012A} are derived using self-similar Schwarzschild coordinates. There are advantages to both approaches, but the latter approach is more useful in the construction of shock-wave solutions, so this approach is adopted. Once Smoller and Temple's system of ODE are introduced, the TOV, FLRW, and asymptotically FLRW solutions are derived in self-similar Schwarzschild coordinates. Proceeding these derivations, the shock wave construction process is introduced, followed by important results related to regularity and the Lax characteristic conditions.

Section \ref{Section3} begins with an alternative derivation of the explicit one-parameter family of general relativistic shock waves originally derived in \cite{ST1995}. This warm-up derivation introduces Lemma \ref{L2}, which is central to the construction of the more general two-parameter family of shock waves, with both parameter counts including $\sigma$. Since the asymptotically Friedmann spacetimes are not known explicitly, numerical approximations are used to construct the two-parameter family of shock waves, with Lemma \ref{L5} from Section \ref{Section5} justifying these approximations. The acceleration parameter and shock position are then approximated in the pure radiation case.

Section \ref{Section4} begins with Lemma \ref{L3}, which generalises the Lax characteristic conditions to a broad family of outgoing general relativistic shock waves. This is then applied to flat FLRW spacetimes as part of Theorem \ref{T4}, which is a refinement of an equivalent theorem from \cite{ST1995}. Finally, an analysis of the Rankine--Hugoniot jump conditions results in Theorem \ref{T6}, which establishes the Lax characteristic conditions of all outgoing spherically symmetric self-similar general relativistic shock waves with a TOV exterior.

Section \ref{Section5} begins by introducing Lemma \ref{L5}, a monotonicity lemma that bounds the family of asymptotically Friedmann spacetimes within a broad region of the ODE system phase space. The remainder of this section is dedicated to the proof of Theorem \ref{T7}, which is a rigorous existence proof of the unique shock wave formed from matching an asymptotically Friedmann spacetime to a TOV spacetime, each with a pure radiation equation of state. Theorem \ref{T7} is the main result of this paper and the first result to establish the existence of a non-explicit global weak solution of the ODE system derived by Smoller and Temple. The proof is complicated by the fact that the matching of the interior and exterior spacetimes occur near the sonic surface of the ODE system phase space.

Finally, Section \ref{Section6} introduces a conjecture regarding the rigorous existence of the full family of shock waves and discusses some open problems that remain.

\section{Preliminaries}\label{Section2}

\subsection{The Spherically Symmetric Self-Similar Field Equations in Schwarzschild Coordinates}\label{Subsection2.1}

Consider first the Einstein field equations
\begin{align}
	G = \kappa T,\label{2.1}
\end{align}
where $G$ is the Einstein curvature tensor, $T$ the stress-energy-momentum tensor and
\begin{align*}
	\kappa = \frac{8\pi\mathcal{G}}{c^4},
\end{align*}
the coupling constant. The coupling constant is comprised of $c$, the speed of light, and $\mathcal{G}$, the gravitational constant. When modelling a perfect fluid, that is, one without shear stresses, viscosity or heat conduction, the stress-energy-momentum tensor takes the form
\begin{align}
	T = \left(\rho+\frac{p}{c^2}\right)\vec{u}\otimes\vec{u} + pg,\label{2.2}
\end{align}
where $g$ is the metric tensor, $\rho$ the fluid mass-energy density, $p$ the fluid pressure and $\vec{u}$ the fluid four-velocity. By construction, the Einstein field equations with a perfect fluid source become the compressible Euler equations in the Newtonian limit $c\to\infty$, and as such, these equations are referred to as the Einstein--Euler equations.
\begin{definition}
	\label{D1}
	A \emph{similarity solution} is a solution that is spherically symmetric and self-similar of the first kind, that is, the self-similar variable takes the form
	\begin{align*}
		\xi = \frac{r}{t},
	\end{align*}
	where $t$ and $r$ are the temporal and radial coordinates respectively.
\end{definition}
From now onwards, the constants $c$ and $\mathcal{G}$ are set to one and all solutions are assumed to be similarity solutions of the Einstein--Euler equations. Any metric of a similarity solution may be written, without loss of generality, in the self-similar Schwarzschild coordinate form
\begin{align}
	ds^2 = -B(\xi)dt^2 + \frac{1}{A(\xi)}dr^2 + r^2d\Omega^2,\label{2.3}
\end{align}
where $A,B>0$ and $d\Omega^2$ denotes the standard metric on the unit two-sphere, that is
\begin{align*}
	d\Omega^2 = d\theta^2 + \sin^2(\theta)d\phi^2.
\end{align*}
Under the assumption of spherical symmetry the fluid four-velocity may also be written without loss of generality as
\begin{align*}
	\vec{u} = (u^0,u^1,0,0).
\end{align*}
Furthermore, under the normalisation condition
\begin{align*}
	g(\vec{u},\vec{u}) = -1,
\end{align*}
the fluid four-velocity has only one independent component. This normalisation condition implies that the fluid four-velocity can be fully specified through the following single variable.
\begin{definition}
	\label{D2}
	The \emph{Schwarzschild coordinate velocity} is defined as
	\begin{align*}
		v = \frac{1}{\sqrt{AB}}\frac{u^1}{u^0}.
	\end{align*}
\end{definition}
Together with $A$, $B$, $\rho$ and $p$, the Schwarzschild coordinate velocity $v$ is one of five unknown variables that completely specify a similarity solution. As there are only four independent equations under all aforementioned assumptions, an equation of state is required to close the system. In this light, we assume that all solutions have a barotropic equation of state, that is, one of the form $p=p(\rho)$. Due to symmetry constraints, it is demonstrated in \cite{CT1971} that all similarity solutions with a barotropic equation of state must take the more restricted linear form
\begin{align}
	p = \sigma\rho\label{2.4}
\end{align}
for some constant $\sigma$. Note that this only applies to spherically symmetric solutions of the self-similar variable $\xi$. Self-similar solutions of the second kind, which include those that admit the self-similar variable
\begin{align*}
	\xi_\lambda = \frac{r}{t^{\lambda}}
\end{align*}
for some non-zero constant $\lambda$, have different restricted forms when admitting a barotropic equation of state. Physically, $\sigma$ represents the square root of the sound speed in the fluid, so for a strictly positive pressure and subluminal sound speed, we require $0<\sigma<1$.
\begin{definition}
	\label{D3}
	The special case $\sigma=\frac{1}{3}$ corresponds to the extreme relativistic limit of free particles and the state of matter known as \emph{pure radiation}.
\end{definition}
A perfect fluid stress-energy-momentum tensor with a pure radiation equation of state is a common model for matter during the Radiation Dominated Epoch of the Early Universe. A pure radiation equation of state is not only physically significant, but also mathematically, as it results in a trace-free stress-energy-momentum tensor.

The similarity assumption decreases the complexity of the Einstein--Euler equations considerably, by reducing them to a system of three nonlinear ODE. Thus by substituting (\ref{2.2}) and (\ref{2.3}) into (\ref{2.1}), Smoller and Temple \cite{ST2012A} demonstrate that the Einstein--Euler equations take the form:
\begin{align}
	\xi\frac{dA}{d\xi} &= -\frac{(3+3\sigma)(1-A)v}{\{\cdot\}_S},\label{2.5}\\
	\xi\frac{dG}{d\xi} &= -G\left[\left(\frac{1-A}{A}\right)\frac{(3+3\sigma)[(1+v^2)G-2v]}{2\{\cdot\}_S}-1\right],\label{2.6}\\
	\xi\frac{dv}{d\xi} &= -\left(\frac{1-v^2}{2\{\cdot\}_D}\right)\left[3\sigma\{\cdot\}_S+\left(\frac{1-A}{A}\right)\frac{(3+3\sigma)^2\{\cdot\}_N}{4\{\cdot\}_S}\right],\label{2.7}
\end{align}
in addition to the constraint
\begin{align}
	\rho = \frac{3(1-v^2)(1-A)G}{\kappa r^2\{\cdot\}_S}.\label{2.8}
\end{align}
The variable $G$, not to be confused with the Einstein tensor, is defined as
\begin{align}
	G = \frac{\xi}{\sqrt{AB}},\label{2.9}
\end{align}
with the bracketed terms defined by:
\begin{align*}
	\{\cdot\}_S &= 3(G-v) - 3\sigma v(1-Gv),\\
	\{\cdot\}_N &= -3(G-v)^2 + 3\sigma v^2(1-Gv)^2,\\
	\{\cdot\}_D &= \frac{3}{4}(3+3\sigma)\left[(G-v)^2-\sigma(1-Gv)^2\right].
\end{align*}
Note that the original derivation, provided in \cite{ST2012A}, is only completed in the pure radiation case. It is not difficult to modify this derivation to yield equations (\ref{2.5})--(\ref{2.8}) for general $\sigma$, but this derivation will not be given here. Note also that under the change of variable
\begin{align}
	\xi = e^s,\label{2.10}
\end{align}
the equations become explicitly autonomous, since
\begin{align}
	\xi\frac{d}{d\xi} = \frac{d}{ds}.\label{2.11}
\end{align}
The autonomous nature of these equations distinguish them from the self-similar ODE derived by Cahill and Taub \cite{CT1971}, which are derived using self-similar comoving coordinates. It is worth noting that in self-similar comoving coordinates, the variable $G$ is the Schwarzschild coordinate analogue of the variable $V$ defined in \cite{CT1971}, which represents surfaces of constant $\xi$ relative to the fluid.

\subsection{Tolman--Oppenheimer--Volkoff Spacetimes}\label{Subsection2.2}

The Tolman--Oppenheimer--Volkoff (TOV) spacetimes are the subfamily of spherically symmetric perfect fluid spacetimes that are static. It is demonstrated in \cite{CT1971} that the self-similar subset of TOV spacetimes which solve the Einstein--Euler equations with a barotropic equation of state form the unique family of static similarity spacetimes. In the context of equations (\ref{2.5})--(\ref{2.8}), these spacetimes are distinguished by having a Schwarzschild coordinate velocity that is identically zero.
\begin{proposition}
	\label{P1}
	Similarity solutions of the Einstein--Euler equations are static if and only if the Schwarzschild coordinate velocity is identically zero.
\end{proposition}
\begin{proof}
	The proof is given in Appendix \ref{A.1}.
\end{proof}
The TOV spacetimes are a remarkably convenient and simple set of solutions to the Einstein--Euler equations because they can be placed in a coordinate system which is comoving, explicitly self-similar and in Schwarzschild form simultaneously.
\begin{definition}
	\label{D4}
	A \emph{scale transformation} is a coordinate transformation of the form:
	\begin{align*}
		\bar{t} &= \mathscr{T}_0t,\\
		\bar{r} &= \mathscr{R}_0r,
	\end{align*}
	where $\mathscr{T}_0$ and $\mathscr{R}_0$ are constants. A parameter that appears in a similarity solution is \emph{essential} if it cannot be removed by a scale transformation and \emph{inessential} if it can.
\end{definition}
Essential and inessential parameters are discussed in more detail in \cite{CT1971}. It is worth noting that when counting the number of essential parameters in a solution, typically a value for $\sigma$ is first fixed and the remaining essential parameters are counted. This convention is not adopted in this paper for reasons that will become apparent later, instead, it is made clear whether the equation of state parameter is included in each instance.
\begin{proposition}
	\label{P2}
	The one-parameter family of TOV spacetimes, denoted by TOV$(\sigma)$, are given in self-similar comoving Schwarzschild coordinates as:
	\begin{align*}
		ds^2 &= -\alpha^2\xi^{\frac{4\sigma}{1+\sigma}}dt^2 + \frac{1}{1-2M(\sigma)}dr^2 + r^2d\Omega^2,\\
		\rho &= \frac{2M(\sigma)}{\kappa r^2},\\
		p &= \sigma\rho,
	\end{align*}
	where $\alpha$ is an inessential parameter and
	\begin{align*}
		M(\sigma) = \frac{2\sigma}{1+6\sigma+\sigma^2}.
	\end{align*}
\end{proposition}
\begin{proof}
	This follows from the proof of Proposition \ref{P1}.
\end{proof}

\subsection{Friedmann--Lemaître--Robertson--Walker Spacetimes}\label{Subsection2.3}

The Friedmann--Lemaître--Robertson--Walker (FLRW) spacetimes are the family of spherically symmetric spatially homogeneous spacetimes. Following the notation of \cite{CY1990}, the flat subset of FLRW spacetimes which solve the Einstein--Euler equations with a barotropic equation of state take the following form in self-similar comoving coordinates:
\begin{align*}
	d\hat{s}^2 &= -e^{2\varphi}d\hat{t}^2 + e^{2\psi}d\hat{r}^2 + \mathscr{R}^2\hat{r}^2d\Omega^2,\\
	\rho &= \frac{2\hat{\xi}^2}{\kappa\hat{r}^2},\\
	p &= \sigma\rho,
\end{align*}
where:
\begin{align*}
	e^{2\varphi} &= \beta^2, &
	e^{2\psi} &= \gamma^{-2}\hat{\xi}^{-\frac{4}{3+3\sigma}}, &
	\mathscr{R}^2 &= \hat{\xi}^{-\frac{4}{3+3\sigma}},
\end{align*}
and:
\begin{align*}
	\beta &= \frac{\sqrt{6}}{3+3\sigma}, & \gamma &= \frac{3+3\sigma}{1+3\sigma}.
\end{align*}
Note that the density is independent of $r$ and that the metric can also be put into an explicitly spatially homogeneous form through the purely radial transformation:
\begin{align*}
	\tilde{t} &= \hat{t},\\
	\tilde{r} &= \hat{r}^\frac{1+3\sigma}{3+3\sigma},
\end{align*}
to yield
\begin{align*}
	d\tilde{s}^2 = -\beta^2d\tilde{t}^2 + \tilde{t}^{\frac{4}{3+3\sigma}}\left(d\tilde{r}^2+\tilde{r}^2d\Omega^2\right).
\end{align*}
\begin{definition}
	\label{D5}
	A \emph{self-similar comoving coordinate system} is a comoving coordinate system in which the metric coefficients are written in terms of a single self-similar variable. Similarly, a \emph{self-similar Schwarzschild coordinate system} is a Schwarzschild coordinate system in which the metric coefficients are written in terms of a single self-similar variable. Note that in spherical symmetry, the $r^2$ part of the $r^2d\Omega^2$ two-form is not considered as part of its coefficient for the purpose of this definition.
\end{definition}
\begin{proposition}
	\label{P3}
	The one-parameter family of similarity perfect fluid FLRW spacetimes with barotropic equations of state, denoted by FLRW$(\sigma,1)$, are given in self-similar Schwarzschild coordinates as:
	\begin{align*}
		ds^2 &= -\delta^{-2}\left[1+\frac{1}{3}(1+3\sigma)\hat{\xi}^{\frac{2+6\sigma}{3+3\sigma}}\right]^{-\frac{1-3\sigma}{1+3\sigma}}\left[1-\frac{2}{3}\hat{\xi}^{\frac{2+6\sigma}{3+3\sigma}}\right]^{-1}dt^2 + \left[1-\frac{2}{3}\hat{\xi}^{\frac{2+6\sigma}{3+3\sigma}}\right]^{-1}dr^2 + r^2d\Omega^2,\\
		v &= \frac{2}{\sqrt{6}}\hat{\xi}^{\frac{1+3\sigma}{3+3\sigma}},\\
		\rho &= \frac{3v^2}{\kappa r^2},\\
		p &= \sigma\rho,
	\end{align*}
	where $\delta$ is an inessential parameter and
	\begin{align}
		\xi = \frac{1}{\sqrt{6}}\delta^{-1}(3+3\sigma)\hat{\xi}^{\frac{1+3\sigma}{3+3\sigma}}\left[1+\frac{1}{3}(1+3\sigma)\hat{\xi}^{\frac{2+6\sigma}{3+3\sigma}}\right]^{-\frac{3+3\sigma}{2+6\sigma}}.\label{2.12}
	\end{align}
\end{proposition}
The one in the second argument of FLRW$(\sigma,1)$ corresponds to the subfamily of unperturbed FLRW spacetimes, with the perturbed FLRW spacetimes defined in the next subsection. Note that the $k<0$ and $k>0$ FLRW spacetimes are not self-similar, so only a self-similar perturbation of the $k=0$ FLRW spacetimes will result in a self-similar spacetime.
\begin{proof}
	The proof is given in Appendix \ref{A.1}.
\end{proof}
Spacetimes that solve equations (\ref{2.5})--(\ref{2.7}) can be denoted by the triple $(A,G,v)$, which specifies the metric through $A$ and $G$, the fluid four-velocity through $v$ and the density through constraint (\ref{2.8}).
\begin{proposition}
	\label{P4}
	FLRW$(\sigma,1)$ is given implicitly by:
	\begin{align}
		A &= 1 - v^2,\label{2.17}\\
		G &= \frac{1}{2}(3+3\sigma)v\left(1+\frac{1}{2}(1+3\sigma)v^2\right)^{-1},\label{2.18}\\
		v &= \frac{2}{\sqrt{6}}\hat{\xi}^{\frac{1+3\sigma}{3+3\sigma}}.\label{2.19}
	\end{align}
\end{proposition}
\begin{proof}
	The proof is given in Appendix \ref{A.1}.
\end{proof}
\begin{corollary}
	\label{C1}
	FLRW$(\frac{1}{3},1)$ is given in self-similar Schwarzschild coordinates as:
	\begin{align*}
		ds^2 &= \frac{1}{1-v^2}\left(-\delta^{-2}dt^2+dr^2\right) + r^2d\Omega^2,\\
		v &= \frac{1-\sqrt{1-\delta^2\xi^2}}{\delta\xi},\\
		\rho &= \frac{3v^2}{\kappa r^2},\\
		p &= \sigma\rho,
	\end{align*}
	where $\delta$ is an inessential parameter.
\end{corollary}
\begin{proof}
	Using Proposition \ref{P3} in the case $\sigma=\frac{1}{3}$, relation (\ref{2.12}) can be inverted to yield $v$. The metric then follows from using this inversion and some algebraic manipulation.
\end{proof}
\begin{corollary}
	\label{C2}
	FLRW$(\frac{1}{3},1)$ is given implicitly by:
	\begin{align}
		A &= 1 - v^2,\label{2.20}\\
		G &= \frac{2v}{1+v^2},\label{2.21}\\
		G &= \delta\xi.\label{2.22}
	\end{align}
\end{corollary}
\begin{proof}
	Relations (\ref{2.20}) and (\ref{2.21}) follow immediately from Proposition \ref{P4}. Relation (\ref{2.22}) follows from (\ref{2.9}) and the identity $AB=\delta^{-2}$ from Corollary \ref{C1}.
\end{proof}

\subsection{Asymptotically Friedmann--Lemaître--Robertson--Walker Spacetimes}\label{Subsection2.4}

Let $(A,G,v)$ denote a solution of equations (\ref{2.5})--(\ref{2.7}). Since these equations are autonomous, solutions can be represented by non-intersecting trajectories in $(A,G,v)$ space. The FLRW$(\sigma,1)$ spacetimes solve equations (\ref{2.5})--(\ref{2.7}) and constraint (\ref{2.8}) with their trajectories satisfying
\begin{align*}
	\lim_{\xi\to0}(A,G,v) = (1,0,0).
\end{align*}
The nature of equations (\ref{2.5})--(\ref{2.7}) suggests that to analyse this point, it is helpful to rewrite these equations as functions of $v$, $A$ and $H$, with $H$ defined as the ratio
\begin{align*}
	H = \frac{G}{v}.
\end{align*}
This is completed in \cite{ST2012A} for $\sigma=\frac{1}{3}$, however it is not difficult to reproduce these equations for general $\sigma$, especially when working from equations (\ref{2.5})--(\ref{2.7}), although this will not be done here. Recalling (\ref{2.10}) and (\ref{2.11}), equations (\ref{2.5})--(\ref{2.8}) are given in autonomous form as functions of $v$, $A$ and $H$ as so:
\begin{align}
	\frac{dv}{ds} &= -v\left(\frac{1-v^2}{2\{\cdot\}_D}\right)\left[3\sigma\{\cdot\}_S^*+\left(\frac{1-A}{A}\right)\frac{(3+3\sigma)^2\{\cdot\}_N^*}{4\{\cdot\}_S^*}\right]\label{2.23},\\
	\frac{dA}{ds} &= -\frac{(3+3\sigma)(1-A)}{\{\cdot\}_S^*}\label{2.24},\\
	\frac{dH}{ds} &= -H\left[\left(\frac{1-A}{A}\right)\frac{(3+3\sigma)[(1+v^2)H-2]}{2\{\cdot\}_S^*}-1\right] - \frac{H}{v}\frac{dv}{ds}\label{2.25},
\end{align}
with
\begin{align}
	\kappa\rho r^2 = \frac{3(1-v^2)(1-A)H}{\{\cdot\}_S^*},\label{2.26}
\end{align}
and where:
\begin{align*}
	\{\cdot\}_S^* &= -(3+3\sigma) + (3+3\sigma v^2)H,\\
	\{\cdot\}_N^* &= -(3-3\sigma) + 2(3-3\sigma v^2)H - (3-3\sigma v^4)H^2,\\
	\{\cdot\}_D &= -\frac{1}{4}(3+3\sigma)(3\sigma-3v^2) - \frac{3}{2}(3-3\sigma^2)Hv^2 + \frac{1}{4}(3+3\sigma)(3-3\sigma v^2)H^2v^2.
\end{align*}
In variables $v$, $A$ and $H$, the FLRW$(\sigma,1)$ spacetimes satisfy
\begin{align*}
	\lim_{\xi\to0}(v,A,H) = \left(0,1,\frac{1}{2}(3+3\sigma)\right),
\end{align*}
and it is straightforward to check this is a fixed point of the system (\ref{2.23})--(\ref{2.25}). Following \cite{ST2012A}, a linear analysis of this fixed point is achieved by first representing equations (\ref{2.23})--(\ref{2.25}) as:
\begin{align*}
	v' &= F_1(v,A,H),\\
	A' &= F_2(v,A,H),\\
	H' &= F_3(v,A,H),\\
\end{align*}
which take the vector form
\begin{align*}
	\boldsymbol{U}' = \boldsymbol{F}(\boldsymbol{U}),
\end{align*}
with:
\begin{align*}
	\boldsymbol{U} &= (v,A,H)^T, &
	\boldsymbol{F} &= \big(F_1(\boldsymbol{U}),F_2(\boldsymbol{U}),F_3(\boldsymbol{U})\big)^T.
\end{align*}
Next, the Jacobian of $\boldsymbol{F}$ at the fixed point is calculated. Note that the Jacobian calculated in \cite{ST2012A} contains an error, so a brief new derivation will be produced. To begin, denote the fixed point by $\boldsymbol{U}_0$ and note
\begin{align*}
	dF_2(\boldsymbol{U}_0) = \left(\frac{\partial F_2}{\partial v},\frac{\partial F_2}{\partial A},\frac{\partial F_2}{\partial H}\right)\bigg\vert_{\boldsymbol{U}_0} = (0,2,0).
\end{align*}
Neglecting terms second order in $v$ and second order in terms that vanish at $\boldsymbol{U}_0$ on the right hand side of (\ref{2.23}) gives
\begin{align*}
	dF_1(\boldsymbol{U}_0) = d\left[-v\left(-\frac{2}{3\sigma(3+3\sigma)}\right)\big[9\sigma H-3\sigma(3+3\sigma)\big]\right]_{\boldsymbol{U}_0} = (1,0,0),
\end{align*}
and similarly for (\ref{2.25}),
\begin{align*}
	dF_3(\boldsymbol{U}_0) &= d\left[H-H\left(\frac{1-A}{A}\right)\frac{(3+3\sigma)(H-2)}{2[3H-(3+3\sigma)]}\right]_{\boldsymbol{U}_0}\\
	&+ d\left[H\left(-\frac{2}{3\sigma(3+3\sigma)}\right)\left[9\sigma H-3\sigma(3+3\sigma)+\left(\frac{1-A}{A}\right)\frac{(3+3\sigma)^2[(3\sigma-3)+6H-3H^2]}{4[3H-(3+3\sigma)]}\right]\right]_{\boldsymbol{U}_0}\\
	&= d\left[3H-\frac{6H^2}{3+3\sigma}+\left(\frac{1-A}{A}\right)\frac{(3+3\sigma)(1-H)H}{6\sigma}\right]_{\boldsymbol{U}_0}\\
	&= \left(0,\left(-\frac{1}{A^2}\right)\frac{(3+3\sigma)(1-H)H}{6\sigma},3-\frac{12H}{3+3\sigma}\right)\bigg\vert_{\boldsymbol{U}_0}\\
	&= \left(0,-\frac{(1+3\sigma)(3+3\sigma)^2}{24\sigma},-3\right).
\end{align*}
Thus the Jacobian is given by
\begin{align*}
	d\boldsymbol{F}(\boldsymbol{U}_0) = \left(\begin{array}{c}
		dF_1(\boldsymbol{U}_0)\\
		dF_2(\boldsymbol{U}_0)\\
		dF_3(\boldsymbol{U}_0)
	\end{array}\right)
	= \left(\begin{array}{ccc}
		1 & 0 & 0\\
		0 & 2 & 0\\
		0 & N(\sigma) & -3
	\end{array}\right),
\end{align*}
where
\begin{align*}
	N(\sigma) = -\frac{(1+3\sigma)(3+3\sigma)^2}{24\sigma}.
\end{align*}
This means $\boldsymbol{U}_0$ is a hyperbolic rest point of the system (\ref{2.23})--(\ref{2.25}) with eigenvalues:
\begin{align*}
	\lambda_1 &= 1, & \lambda_2 &= 2, & \lambda_3 &= -3.
\end{align*}
Therefore solutions
\begin{align*}
	\boldsymbol{U}(s) = \boldsymbol{U}_0 + \boldsymbol{V}(s),
\end{align*}
where $\boldsymbol{V}(s)$ solves the linearised equations
\begin{align*}
	\boldsymbol{V}' = d\boldsymbol{F}(\boldsymbol{U}_0)\cdot\boldsymbol{V},
\end{align*}
lie in the two-dimensional unstable manifold $\mathcal{M}_0$ of $\boldsymbol{U}_0$, given by
\begin{align*}
	\mathcal{M}_0 = \left(\begin{array}{c}
		0\\
		1\\
		\frac{1}{2}(3+3\sigma)\\
	\end{array}\right) + \text{Span}\left\lbrace\left(\begin{array}{c}
		1\\
		0\\
		0\\
	\end{array}\right)e^s + \left(\begin{array}{c}
		0\\
		1\\
		0\\
	\end{array}\right)e^{2s}\right\rbrace.
\end{align*}
In particular,
\begin{align*}
	\boldsymbol{U}(s) = \left(\begin{array}{c}
		C_4e^s\\
		1 + C_5e^{2s}\\
		\frac{1}{2}(3+3\sigma)\\
	\end{array}\right)
\end{align*}
for arbitrary constants $C_4$ and $C_5$. In the variable $\xi$, the solutions are given by:
\begin{align}
	A_1(\xi) &= 1 + C_5\xi^2,\label{2.27}\\
	G_1(\xi) &= \frac{1}{2}(3+3\sigma)C_4\xi,\label{2.28}\\
	v_1(\xi) &= C_4\xi,\label{2.29}
\end{align}
with the subscript denoting the fact that $(A_1,G_1,v_1)$ represents a solution to the linearised version of equations (\ref{2.5})--(\ref{2.7}). Now functions $A$, $G$ and $v$ of FLRW$(\sigma,1)$ are given to leading order as:
\begin{align}
	A(\xi) &= 1 - \frac{4}{(3+3\sigma)^2}\delta^2\xi^2 + O_{\xi\to0}(\xi^4),\label{2.30}\\
	G(\xi) &= \delta\xi + O_{\xi\to0}(\xi^3),\label{2.31}\\
	v(\xi) &= \frac{2}{3+3\sigma}\delta\xi + O_{\xi\to0}(\xi^3),\label{2.32}
\end{align}
where we note from Corollary \ref{C2} that all higher order coefficients of $G$ are zero for $\sigma=\frac{1}{3}$. Comparing (\ref{2.27})--(\ref{2.29}) to (\ref{2.30})--(\ref{2.32}) suggests setting $C_4$ and $C_5$, without loss of generality, as
\begin{align*}
	C_4 &= \frac{2}{3+3\sigma}\delta, &
	C_5 &= -\frac{4}{(3+3\sigma)^2}\delta^2a^2,
\end{align*}
where $a$ is an essential parameter. Including $\sigma$, we have that (\ref{2.27})--(\ref{2.29}) is a two-parameter family of solutions originating from the fixed point $\boldsymbol{U}_0$, with the leading order approximations of FLRW$(\sigma,1)$ as a one-parameter subset. The FLRW$(\sigma,1)$ spacetimes correspond to $a=1$ and any other value of $a$ represents a self-similar perturbation from FLRW$(\sigma,1)$. From this point onwards the value of $\delta$ is fixed as
\begin{align}
	\delta = \frac{1}{4}(3+3\sigma)\label{2.33}
\end{align}
so as to simplify calculations and match the notation used in \cite{ST2012A}.
\begin{definition}
	\label{D6}
	The \emph{asymptotically Friedmann} spacetimes, denoted by FLRW$(\sigma,a)$, are defined as the two-parameter family of solutions to (\ref{2.5})--(\ref{2.8}) with the leading order form:
	\begin{align*}
		A(\xi) &= 1 - \frac{1}{4}a^2\xi^2 + O_{\xi\to0}(\xi^4),\\
		G(\xi) &= \frac{1}{4}(3+3\sigma)\xi + O_{\xi\to0}(\xi^3),\\
		v(\xi) &= \frac{1}{2}\xi + O_{\xi\to0}(\xi^3),
	\end{align*}
	where we note that the third-order coefficient of $G$ is zero for $\sigma=\frac{1}{3}$. Furthermore, $a$ is referred to as the \emph{acceleration parameter}.
\end{definition}
The asymptotic form of the FLRW$(\sigma,a)$ spacetimes was first found by Carr and Yahil \cite{CY1990}, with this asymptotic form given in comoving coordinates. The FLRW$(\sigma,a)$ spacetimes are exact solutions of Einstein's field equations, even though they are not known explicitly. Despite this, we can still provide the leading order approximation of FLRW$(\sigma,a)$ solutions local to the centre of expansion.
\begin{proposition}
	\label{P5}
	The FLRW$(\sigma,a)$ spacetimes are given in self-similar Schwarzschild coordinates to leading order as so:
	\begin{align*}
		ds^2 &= -\frac{16}{(3+3\sigma)^2}\left(1+\frac{1}{4}a^2\xi^2\right)dt^2 + \left(1+\frac{1}{4}a^2\xi^2\right)dr^2 + r^2d\Omega^2 + O_{\xi\to0}(\xi^4),\\
		v &= \frac{1}{2}\xi + O_{\xi\to0}(\xi^3),\\
		\rho &= \frac{3a^2\xi^2}{4\kappa r^2} + O_{\xi\to0}(\xi^4),\\
		p &= \sigma\rho.
	\end{align*}
\end{proposition}
\begin{proof}
	This follows from Proposition \ref{P3} and Definition \ref{D6} by noting that $ABG^2=\xi^2$.
\end{proof}

The closer the acceleration parameter is to one, and the closer an observer is to the centre of expansion, the closer the spacetime is to spatial homogeneity. Regardless, any FLRW$(\sigma,a)$ spacetime with $a\neq1$ is still inhomogeneous, and thus placing our galaxy anywhere within this spacetime would violate the Copernican principle. As the name suggests, the acceleration parameter is also responsible for the rate of accelerated expansion of the spacetime. For an observer near the centre of expansion, and an acceleration parameter only slightly larger than one, an FLRW$(\frac{1}{3},a)$ universe would appear similar to the Standard Model of Cosmology in the Radiation Dominated Epoch, that is, it would appear flat and with an accelerated expansion similar to one modelled by a small positive cosmological constant. Smoller and Temple give a more in-depth consideration of the cosmological implications of FLRW$(\sigma,a)$ spacetimes in \cite{ST2012A}.

\subsection{Shock Wave Construction}\label{Subsection2.5}

Suppose that we have two spherically symmetric, although not necessarily self-similar, solutions to the Einstein--Euler equations. Let us denote these solutions by the triples $(g,\rho,\vec{u})$ and $(\bar{g},\bar{\rho},\vec{v})$, and assume that these solutions have barotropic equations of state $p=p(\rho)$ and $\bar{p}=\bar{p}(\bar{\rho})$ respectively. Since we are assuming spherical symmetry, when specifying a set of coordinates $(t,r,\theta,\phi)$, it is sufficient to only consider the coordinates $(t,r)$. In this light, let metrics $g$ and $\bar{g}$ be given in Schwarzschild coordinates $(t,r)$ and $(\bar{t},\bar{r})$ respectively as so:
\begin{align*}
	ds^2 &= -B(t,r)dt^2 + \frac{1}{A(t,r)}dr^2 + r^2d\Omega^2,\\
	d\bar{s}^2 &= -\bar{B}(\bar{t},\bar{r})d\bar{t}^2 + \frac{1}{\bar{A}(\bar{t},\bar{r})}d\bar{r}^2 + \bar{r}^2d\Omega^2,
\end{align*}
where the angular coordinates $(\theta,\phi)$ and $(\bar{\theta},\bar{\phi})$ have been identified.
\begin{definition}
	\label{D7}
	We say that two metrics can be \emph{matched} on a spherical surface $\tilde{r}=\Phi(\tilde{t})$ if there exists a common set of coordinates $(\tilde{t},\tilde{r})$ such that the coefficients of the metrics agree on this surface when written in these coordinates.
\end{definition}
It is not required that the metrics be given in Schwarzschild coordinates in order to be matched, but it does provide a convenient set of coordinates from which the metrics can be compared. For metrics $g$ and $\bar{g}$, we may simply take $(t,r)$ as our common set of coordinates and ask which transformation of the form:
\begin{align*}
	\bar{t} &= \bar{t}(t,r),\\
	\bar{r} &= \bar{r}(t,r),
\end{align*}
is required in order to match these metrics. The reason Schwarzschild coordinates are preferred is because the $d\Omega^2$ coefficients are automatically matched through the identification $\bar{r}=r$. In order to avoid introducing $dtdr$ terms, the identification implies the most general transformation that can be applied takes the form $\bar{t}=\bar{t}(t)$. Thus for two metrics given in Schwarzschild coordinates, the process of matching these metrics reduces to the existence of a spherical surface $r=\Phi(t)$ and a coordinate transformation $\bar{t}=\bar{t}(t)$ that satisfy the algebraic-differential equations:
\begin{align*}
	B(t,\Phi(t)) &= \bar{B}(\bar{t}(t),\Phi(t))[\bar{t}'(t)]^2,\\
	A(t,\Phi(t)) &= \bar{A}(\bar{t}(t),\Phi(t)).
\end{align*}
If these equations are soluble, then metrics $g$ and $\bar{g}$ can be matched on the surface $r=\Phi(t)$. However, such a matching does not automatically imply that mass-energy and momentum are conserved across the surface. With this in mind, let us assume that there exists a set of coordinates $(t,r)$ for which the metrics match on the spherical surface $r=\Phi(t)$ and define
\begin{align*}
	\Sigma = \{(t,r,\theta,\phi):r=\Phi(t),t>0\}.
\end{align*}
\begin{definition}
	\label{D8}
	We say that the spacetime given by the matched metric $g\cup\bar{g}$, along with the associated hydrodynamic variables, forms a \emph{shock-wave solution} of the Einstein--Euler equations if the Rankine--Hugoniot jump conditions hold across the surface $\Sigma$. Furthermore, the spherical surface $\Sigma$ is known as the \emph{shock surface} or simply the \emph{shock}.
\end{definition}
As like in classical shock-wave theory, the Rankine--Hugoniot jump conditions express the weak form of the conservation of mass-energy and momentum across the shock-surface.
\begin{proposition}
	\label{P6}
	Let $\vec{p}\in\Sigma$ and $U$ be a neighbourhood of $\vec{p}$, then the weak form of the conservation of mass-energy and momentum across $\Sigma\cap U$ is given by:
	\begin{align}
		\int_U T^{\mu\nu}\nabla_\nu\varphi\, d\vec{x} = 0\quad \forall\ \varphi \in C^\infty_c(U),\label{2.34}
	\end{align}
	or equivalently:
	\begin{align}
		\int_U G^{\mu\nu}\nabla_\nu\varphi\, d\vec{x} = 0\quad \forall\ \varphi \in C^\infty_c(U).\label{2.35}
	\end{align}
\end{proposition}
\begin{proof}
	The proof is given in Appendix \ref{A.1}.
\end{proof}
The following proposition specifies the general relativistic form of the Rankine--Hugoniot jump conditions.
\begin{proposition}
	\label{P7}
	The Rankine--Hugoniot jump conditions are given by
	\begin{align*}
		[G^{\mu\nu}]n_\nu = 0,
	\end{align*}
	where $\vec{n}$ is the normal vector to $\Sigma$ and
	\begin{align*}
		[G^{\mu\nu}] = G^{\mu\nu}(g) - G^{\mu\nu}(\bar{g})
	\end{align*}
	is the jump in $G$ across $\Sigma$.
\end{proposition}
\begin{proof}
	The proof is given in Appendix \ref{A.1}.
\end{proof}

\subsection{Regularity}\label{Subsection2.6}

As like in Subsection \ref{Subsection2.5}, consider the solution triples $(g,\rho,\vec{u})$ and $(\bar{g},\bar{\rho},\vec{v})$. Assume that these solutions can be matched Lipschitz continuously on a spherical surface $\Sigma$ with a spacelike normal vector $\vec{n}$ to form the matched metric $g\cup\bar{g}$. Furthermore, let $g\cup\bar{g}$ satisfy the Rankine--Hugoniot jump condition across $\Sigma$ so that $g\cup\bar{g}$ forms a shock-wave solution. For the rest of this subsection, the matched metric is referred to simply as the metric.

It is reasonable to be concerned with the regularity of such a solution, since a Lipschitz continuous shock wave has discontinuous first-order derivatives, and thus delta function sources in the second-order derivatives. Given that the Einstein tensor comprises second-order derivatives of the metric, this too is expected to harbour delta function sources. On the right side of the Einstein field equations, the hydrodynamic variables $\rho$, $p$ and $\vec{u}$, along with the metric, form the stress-energy-momentum tensor, and since the hydrodynamic variables are expected to be at worst discontinuous across the shock, so too is the stress-energy-momentum tensor. This is problematic, since the Einstein field equations cannot have different levels of regularity on the left and right hand sides of the equations. However, it turns out that even though delta function sources may appear in the second-order derivatives of the metric at the shock, with these being coordinate dependent, the Einstein tensor does not have any delta function sources, that is, the delta function sources cancel in the Einstein tensor. This miraculous result, from \cite{ST1995}, is summarised in the following theorem.
\begin{theorem}
	\label{T1}
	Let $\Sigma$ denote a smooth, three-dimensional surface with a spacelike normal vector $\vec{n}$. Assume that the components of the metric are continuous on each side of $\Sigma$ and Lipschitz continuous across $\Sigma$ in some fixed coordinate system. Then the following statements are equivalent:
	\begin{enumerate}
		\item $[K]=0$ at each point of $\Sigma$, where $K$ is the second fundamental form of the metric.
		\item The Riemann curvature and Einstein tensors, viewed as second-order operators on the metric components, produce no delta function sources on $\Sigma$.
		\item For each point $\vec{p}\in\Sigma$ there exists a $C^{1,1}$ coordinate transformation defined in a neighbourhood of $\vec{p}$ such that in the new coordinates, which can be taken to be the Gaussian normal coordinates for the surface, the metric components are $C^{1,1}$ functions of these coordinates.
		\item For each point $\vec{p}\in\Sigma$ there exists a coordinate frame that is locally Lorentzian at $\vec{p}$ and can be reached from the original coordinates by a $C^{1,1}$ coordinate transformation.
	\end{enumerate}
	Moreover, if any one of these statements hold, then the Rankine--Hugoniot jump conditions
	\begin{align*}
		[G^{\mu\nu}]n_\nu = 0
	\end{align*}
	hold at each point of $\Sigma$.
\end{theorem}
This theorem provides a criterion for the removal of the delta function sources and also a coordinate system for which the shock-wave solution can achieve optimal regularity, that is, when the metric has a Lipschitz continuous derivative at the shock. The following theorem, also from \cite{ST1995}, provides convenient criteria for satisfying one of the equivalent statements of Theorem \ref{T1}.
\begin{theorem}
	\label{T2}
	Assume the following:
	\begin{enumerate}
		\item That $g$ and $\bar{g}$ are two spherically symmetric metrics that match across a three-dimensional surface $\Sigma$ to form the matched metric $g\cup\bar{g}$.
		\item The matched metric is Lipschitz continuous cross $\Sigma$.
		\item The normal $\vec{n}$ to $\Sigma$ is non-null.
	\end{enumerate}
	Then the following are equivalent:
	\begin{enumerate}
		\item $[G^{\mu\nu}]n_\nu=0$.
		\item $[G^{\mu\nu}]n_\mu n_\nu=0$.
		\item $[K]=0$ at each point of $\Sigma$, where $K$ is the second fundamental form of the metric.
		\item The components of the matched metric in any Gaussian-normal coordinate system are $C^{1,1}$ functions of these coordinates across $\Sigma$.
	\end{enumerate}
\end{theorem}
If the conditions of Theorem \ref{T2} are satisfied, then it is clear that the weak form of mass-energy and momentum conservation across the shock surface is equivalent to the single condition
\begin{align*}
	[T^{\mu\nu}]n_\mu n_\nu = 0.
\end{align*}
Thus the spherically symmetric Rankine--Hugoniot jump conditions reduce to the single equivalent condition
\begin{align*}
	[G^{\mu\nu}]n_\mu n_\nu = 0.
\end{align*}
Therefore a shock-wave solution, which satisfies the Rankine--Hugoniot jump conditions by definition, only requires the metric to be continuous on each side of $\Sigma$ and Lipschitz continuous across $\Sigma$ to satisfy the equivalent statements of Theorems \ref{T1} and \ref{T2}. The proofs of these theorems are given in \cite{ST1994}.

\subsection{The Lax Characteristic Conditions}\label{Subsection2.7}

The Lax characteristic conditions, when satisfied, guarantee characteristics in the same family as the shock impinge on the shock from both sides \cite{L1957,S1994}. This ensures the shock surface does not dissipate, providing the shock with a degree of stability, that is, the shock-wave solution is expected to remain as a shock-wave solution. Note that the Lax characteristic conditions lead to time irreversibility, since characteristics impinge on the shock, entropy increases and information is lost. In classical gas dynamics, the density and pressure are always larger behind stable shock waves, which means spherically symmetric shock waves with a greater pressure and density on the interior are expected to expand.

Consider again the solution triples $(g,\rho,\vec{u})$ and $(\bar{g},\bar{\rho},\vec{v})$ with equations of state $p=p(\rho)$ and $\bar{p}=\bar{p}(\bar{\rho})$ respectively, and assume that these solutions form the shock-wave solution $g\cup\bar{g}$. As a spherical surface has an interior and exterior, let $g$ represent the interior metric, which is given in comoving coordinates $(\hat{t},\hat{r})$ as
\begin{align*}
	d\hat{s}^2 = -e^{2\varphi}d\hat{t}^2 + e^{2\psi}d\hat{r}^2 + \mathscr{R}^2\hat{r}^2d\Omega^2.
\end{align*}
Finally, let $\bar{g}$ represent the exterior metric, with the associated comoving coordinates denoted by $(\bar{t},\bar{r})$. The objective is to determine the Lax characteristic conditions at the shock surface.
\begin{lemma}
	\label{L1}
	The shock speed relative to the interior fluid is given by
	\begin{align}
		e^{\psi-\varphi}\dot{\Phi},\label{2.36}
	\end{align}
	where $\hat{r}=\Phi(\hat{t})$ is the position of the shock in coordinates comoving with the interior fluid.
\end{lemma}
\begin{proof}
	The proof is given in Appendix \ref{A.1}.
\end{proof}
Let $\tilde{\lambda}^+_{Int}$ and $\tilde{\lambda}^-_{Int}$ denote the speeds of the interior characteristics in $(\tilde{t},\tilde{r})$ coordinates. Since the characteristic speeds on the interior side of the shock equal the sound speeds in locally Minkowskian coordinates, we have
\begin{align*}
	\tilde{\lambda}^\pm_{Int} = \pm\sqrt{\frac{dp}{d\rho}}.
\end{align*}
The Lax characteristic conditions determine if the characteristic curves in the family of the shock impinge upon the shock from both sides. Since we are considering shocks that are outward moving with respect to $\hat{r}$ and $\bar{r}$, it follows that on the interior side, only the positive characteristic can impinge on the shock \cite{S1994}. Alternatively, the Lax characteristic conditions can be thought of as selecting whether it is the outgoing or ingoing shock that is stable. Let $\tilde{\lambda}^+_{Ext}$ and $\tilde{\lambda}^-_{Ext}$ denote the speeds of the exterior characteristics in $(\tilde{t},\tilde{r})$ coordinates. The Lax characteristic conditions are given as the inequalities
\begin{align}
	\tilde{\lambda}^+_{Ext} < s < \tilde{\lambda}^+_{Int}\label{2.39},
\end{align}
where $s$ is the speed of the shock in $(\tilde{t},\tilde{r})$ coordinates.
\begin{proposition}
	\label{P8}
	For an expanding spherically symmetric general relativistic shock wave, the Lax characteristic conditions are given as the inequalities
	\begin{align}
		\frac{\tilde{w}+\sqrt{\frac{d\bar{p}}{d\bar{\rho}}}}{1+\tilde{w}\sqrt{\frac{d\bar{p}}{d\bar{\rho}}}} < e^{\psi-\varphi}\dot{\Phi} < \sqrt{\frac{dp}{d\rho}},\label{2.40}
	\end{align}
	where
	\begin{align*}
		\tilde{w} = e^{\psi-\varphi}\frac{\partial\hat{r}}{\partial\bar{t}}\left(\frac{\partial\hat{t}}{\partial\bar{t}}\right)^{-1}.
	\end{align*}
\end{proposition}
\begin{proof}
	The proof is given in Appendix \ref{A.1}.
\end{proof}

\section{FLRW--TOV Shock Waves}\label{Section3}

\subsection{Analytical Results}\label{Subsection3.1}

The objective of this section is the construction of the family of FLRW$(\sigma,a)$--TOV$(\bar{\sigma})$ shock waves. As remarked previously, the triple $(A,G,v)$ is used to denote a solution of the spherically symmetric self-similar Einstein--Euler equations (\ref{2.5})--(\ref{2.7}).
\begin{definition}
	\label{D9}
	A shock-wave solution with an FLRW$(\sigma,a)$ spacetime on the interior and a TOV$(\bar{\sigma})$ spacetime on the exterior is referred to as an FLRW$(\sigma,a)$--TOV$(\bar{\sigma})$ shock wave.
\end{definition}
Since all FLRW$(\sigma,a)$--TOV$(\bar{\sigma})$ shock waves share a TOV$(\bar{\sigma})$ exterior, the following lemma is of great utility in their construction.
\begin{lemma}
	\label{L2}
	Let $(A,G,v)$ denote a similarity solution to the Einstein--Euler equations with equation of state $p=\sigma\rho$. If there exists a $\xi_0>0$ such that
	\begin{align}
		A(\xi_0) = 1 - 2M(\bar{\sigma}),\label{3.1}
	\end{align}
	then $(A,G,v)$ can be matched to TOV$(\bar{\sigma})$ on the surface $\xi=\xi_0$ and the Rankine--Hugoniot jump condition is given by
	\begin{align}
		\frac{[\sigma+v^2(\xi_0)]G(\xi_0)-(1+\sigma)G^2(\xi_0)v(\xi_0)}{[1+\sigma v^2(\xi_0)]G(\xi_0)-(1+\sigma)v(\xi_0)} = \bar{\sigma}.\label{3.2}
	\end{align}
\end{lemma}
\begin{proof}
	The metric of the $(A,G,v)$ solution in self-similar Schwarzschild coordinates is given, without loss of generality, by
	\begin{align*}
		ds^2 = -B(\xi)dt^2 + \frac{1}{A(\xi)}dr^2 + r^2d\Omega^2,
	\end{align*}
	and recall by Proposition \ref{P2} that TOV$(\bar{\sigma})$ is given in self-similar comoving Schwarzschild coordinates as:
	\begin{align*}
		d\bar{s}^2 &= -\bar{\xi}^{\frac{4\bar{\sigma}}{1+\bar{\sigma}}}d\bar{t}^2 + \frac{1}{1-2M(\bar{\sigma})}d\bar{r}^2 + \bar{r}^2d\Omega^2,\\
		\bar{\rho} &= \frac{2M(\bar{\sigma})}{\kappa\bar{r}^2},\\
		\bar{p} &= \bar{\sigma}\bar{\rho},
	\end{align*}
	where the inessential parameter has been set to one and
	\begin{align*}
		M(\bar{\sigma}) = \frac{2\bar{\sigma}}{1+6\bar{\sigma}+\bar{\sigma}^2}.
	\end{align*}
	Because both metrics are specified in Schwarzschild coordinates, the $d\Omega^2$ components automatically match under the identification $\bar{r}=r$. Matching the $dr^2$ components implies that the shock surface is defined by $\xi=\xi_0$, with the constant $\xi_0$ given implicitly by (\ref{3.1}). This also implies that $B$ is constant on the surface. A temporal rescaling of the form
	\begin{align*}
		\bar{t} = \alpha t
	\end{align*}
	implies
	\begin{align*}
		\alpha\bar{\xi} = \xi = \xi_0 = \alpha\bar{\xi}_0,
	\end{align*}
	and matches the $dt^2$ coefficients providing $\alpha$ satisfies
	\begin{align}
		\alpha^2\bar{\xi}_0^{\frac{4\bar{\sigma}}{1+\bar{\sigma}}} = B(\xi_0).\label{3.3}
	\end{align}
	With the matching in place, recall from Subsection \ref{Subsection2.5} that the Rankine--Hugoniot jump condition is equivalent to
	\begin{align*}
		[T^{\mu\nu}]n_\mu n_\nu = T^{\mu\nu}(g,\rho,p,\vec{u}_{FLRW})n_\mu n_\nu - T^{\mu\nu}(g,\bar{\rho},\bar{p},\vec{u}_{TOV})n_\mu n_\nu = 0,
	\end{align*}
	where $\vec{n}$ is the outward normal to the shock surface. Using this and $p=\sigma\rho$, we obtain
	\begin{align*}
		(1+\sigma)\rho u_{FLRW}^\mu u_{FLRW}^\nu n_\mu n_\nu + \sigma\rho|\vec{n}|^2 - (1+\bar{\sigma})\bar{\rho}u_{TOV}^\mu u_{TOV}^\nu n_\mu n_\nu - \bar{\sigma}\bar{\rho}|\vec{n}|^2 = 0.
	\end{align*}
	Now since the surface is defined by $\xi=\xi_0$, which is equivalent to
	\begin{align*}
		r - \xi_0t = 0,
	\end{align*}
	then the components of the normal satisfy
	\begin{align*}
		n_\mu dx^\mu = d(r-\xi_0t) = -\xi_0dt + dr.
	\end{align*}
	Given that the metric components are identified on the surface, the following identities are obtained:
	\begin{align*}
		|\vec{n}|^2 &= A(\xi_0)-\xi_0^2B^{-1}(\xi_0),\\
		u_{FLRW}^0 &= [1-v^2(\xi_0)]^{-\frac{1}{2}}B^{-\frac{1}{2}}(\xi_0),\\
		u_{FLRW}^1 &= v(\xi_0)[1-v^2(\xi_0)]^{-\frac{1}{2}}A^{\frac{1}{2}}(\xi_0),\\
		u_{FLRW}^\mu u_{FLRW}^\nu n_\mu n_\nu &= [1-v^2(\xi_0)]^{-1}\left[v(\xi_0)A^{\frac{1}{2}}(\xi_0)-\xi_0B^{-\frac{1}{2}}(\xi_0)\right]^2,\\
		u_{TOV}^\mu u_{TOV}^\nu n_\mu n_\nu &= \xi_0^2B^{-1}(\xi_0).
	\end{align*}
	Applying these identities puts the Rankine--Hugoniot jump condition in the form
	\begin{align*}
		0 &= (1+\sigma)[1-v^2(\xi_0)]^{-1}\left[v(\xi_0)A^{\frac{1}{2}}(\xi_0)-\xi_0B^{-\frac{1}{2}}(\xi_0)\right]^2\rho\\
		&- (1+\bar{\sigma})\xi_0^2B^{-1}(\xi_0)\bar{\rho}\\
		&+ \left[A(\xi_0)-\xi_0^2B^{-1}(\xi_0)\right](\sigma\rho-\bar{\sigma}\bar{\rho}).
	\end{align*}
	Dividing by $A(\xi_0)$ and substituting $B(\xi_0)$ for $G(\xi_0)$ then yields
	\begin{align*}
		0 &= (1+\sigma)[1-v^2(\xi_0)]^{-1}[v(\xi_0)-G(\xi_0)]^2\rho\\
		&- (1+\bar{\sigma})G^2(\xi_0)\bar{\rho}\\
		&+ [1-G^2(\xi_0)](\sigma\rho-\bar{\sigma}\bar{\rho}).
	\end{align*}
	Finally, applying (\ref{2.8}) and (\ref{3.1}) gives (\ref{3.2}), which completes the proof.
\end{proof}
As FLRW$(\sigma,1)$ is known explicitly, it is possible to construct an explicit FLRW$(\sigma,1)$--TOV$(\bar{\sigma})$ shock wave. Such a construction was first achieved by Cahill and Taub in \cite{CT1971} for the case $\sigma=\frac{1}{3}$, and then later in full generality by Smoller and Temple in \cite{ST1995}. Alternatively, the result can be derived directly from Lemma \ref{L2}.
\begin{theorem}
	\label{T3}
	For each $0<\sigma<1$, FLRW$(\sigma,1)$ can be matched to TOV$(\bar{\sigma})$ to form a general relativistic shock wave providing
	\begin{align}
		\bar{\sigma} = H(\sigma),\label{3.4}
	\end{align}
	where
	\begin{align*}
		H(\sigma) = \frac{1}{2}\sqrt{9\sigma^2+54\sigma+49} - \frac{3}{2}\sigma - \frac{7}{2}.
	\end{align*}
\end{theorem}
\begin{proof}
	The matching follows similarly to the matching completed in the proof of Lemma \ref{L2}, but with (\ref{3.1}) and (\ref{3.3}) replaced with
	\begin{align*}
		1-\frac{2}{3}\hat{\xi}_0^{\frac{2+6\sigma}{3+3\sigma}} = 1 - 2M(\bar{\sigma})
	\end{align*}
	and
	\begin{align*}
		\alpha^2\bar{\xi}_0^{\frac{4\bar{\sigma}}{1+\bar{\sigma}}} = \frac{16}{(3+3\sigma)^2}\left[1+\frac{1}{3}(1+3\sigma)\hat{\xi}_0^{\frac{2+6\sigma}{3+3\sigma}}\right]^{-\frac{2-6\sigma}{2+6\sigma}}\left[1-\frac{2}{3}\hat{\xi}_0^{\frac{2+6\sigma}{3+3\sigma}}\right]^{-1}
	\end{align*}
	respectively, where the inessential parameter is set by (\ref{2.33}) and $\hat{\xi}$ is defined implicitly by
	\begin{align*}
		\xi = \frac{4}{\sqrt{6}}\hat{\xi}^{\frac{1+3\sigma}{3+3\sigma}}\left[1+\frac{1}{3}(1+3\sigma)\hat{\xi}^{\frac{2+6\sigma}{3+3\sigma}}\right]^{-\frac{3+3\sigma}{2+6\sigma}}.
	\end{align*}
	Note that this matching is Lipschitz continuous, as any $0<\sigma<1$ and $0<\bar{\sigma}<1$ imply that the components of the interior and exterior metrics are continuous in a neighbourhood of the surface when given in $(t,r)$ coordinates. Thus it remains to show that the condition $\bar{\sigma}=H(\sigma)$ is equivalent to the Rankine--Hugoniot jump condition, which we know by Lemma \ref{L2} is given by
	\begin{align*}
		\frac{[\sigma+v^2(\xi_0)]G(\xi_0)-(1+\sigma)G^2(\xi_0)v(\xi_0)}{[1+\sigma v^2(\xi_0)]G(\xi_0)-(1+\sigma)v(\xi_0)} = \bar{\sigma}.
	\end{align*}
	By Proposition \ref{P4}, $G(\xi_0)$ can be substituted for $v(\xi_0)$, which in turn can be substituted for $A(\xi_0)$ to yield
	\begin{align*}
		\frac{(3\sigma+3[1-A(\xi_0)])(2+(1+3\sigma)[1-A(\xi_0)])-(3+3\sigma)^2[1-A(\xi_0)]}{A(\xi_0)(2+(1+3\sigma)[1-A(\xi_0)])} = \bar{\sigma}.
	\end{align*}
	Finally, substituting $A(\xi_0)$ for $1-2M(\bar{\sigma})$ yields
	\begin{align*}
		\sigma = \frac{\bar{\sigma}(7+\bar{\sigma})}{3(1-\bar{\sigma})},
	\end{align*}
	which is equivalent to $\bar{\sigma}=H(\sigma)$.
\end{proof}
\begin{definition}
	\label{D10}
	The \emph{Rankine--Hugoniot curve}, denoted by
	\begin{align*}
		v = \Gamma_{RH}(G;\sigma,\bar{\sigma}),
	\end{align*}
	is the curve in $(A,G,v)$ space that is generated by constraints (\ref{3.1}) and (\ref{3.2}).
\end{definition}

\subsection{Numerical Results}\label{Subsection3.2}

We are now in a position to extend the family of FLRW$(\sigma,1)$--TOV$(\bar{\sigma})$ shock waves to the family of FLRW$(\sigma,a)$--TOV$(\bar{\sigma})$ shock waves. Even though the FLRW$(\sigma,a)$ spacetimes are exact solutions, they are not known explicitly. One way of describing FLRW$(\sigma,a)$ solutions is to numerically generate their trajectories in $(A,G,v)$ space, such as in Figure \ref{F1}. The statement of Lemma \ref{L5} in Subsection \ref{Subsection5.1}, which determines where FLRW$(\sigma,a)$ trajectories become singular, permits us to have confidence in these numerical approximations away from any singular surfaces. For the definition of a singular surface, see Definition \ref{D11} in Subsection \ref{Subsection4.2}.
\begin{figure}[h]
	\begin{center}
		\includegraphics[width=10cm]{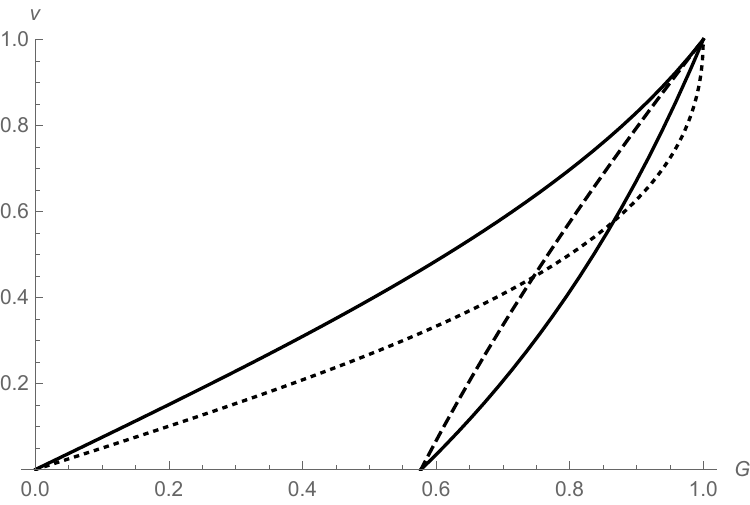}
	\end{center}
	\caption{This figure is a side view of $(A,G,v)$ space and depicts the most important features. The left and right unbroken curves represent the surfaces $\{\cdot\}_S=0$ and $\{\cdot\}_D=0$ respectively. These surfaces have no dependence on $A$ and so remain the same in any constant $A$ plane. The Rankine--Hugoniot curve is represented by the dashed curve and lives in the plane $A=1-2M(\bar{\sigma})$. The dotted curve represents the explicitly known FLRW$(\frac{1}{3},1)$ trajectory.}
	\label{F1}
\end{figure}

The FLRW$(\frac{1}{3},1)$ trajectory obeys the implicit relationship given by Corollary \ref{C2}, that is, as $\xi$ increases from zero, $G$ increases linearly with $\xi$, $v$ increases according to (\ref{2.21}) and $A$ decreases according to (\ref{2.20}). Generic FLRW$(\frac{1}{3},a)$ trajectories are similar to the FLRW$(\frac{1}{3},1)$ trajectory for small $\xi$, but differ as $\xi$ increases. One property that remains similar for larger $\xi$ is the near linear dependence of $G$ on $\xi$. Note that because equations (\ref{2.5})--(\ref{2.7}) are autonomous, all trajectories, the Rankine--Hugoniot curve and surfaces $\{\cdot\}_S=0$ and $\{\cdot\}_D=0$ are all independent of $\xi$. Thus it is often easier to think of $G$ as the independent variable and consider the trajectory, curves and surfaces as a function of $G$.

The TOV$(\bar{\sigma})$ trajectories are easy to represent in $(A,G,v)$ solution space, as they are simply the lines defined by $A=1-2M(\bar{\sigma})$ and $v=0$. Now because
\begin{align}
	\min_{0\leq\bar{\sigma}\leq1}\{1-2M(\bar{\sigma})\} = \frac{1}{2},\label{3.5}
\end{align}
the family of TOV$(\bar{\sigma})$ trajectories span the surface defined by $\frac{1}{2}<A<1$ and $v=0$.

Another reason considering solutions in $(A,G,v)$ space is so useful, is because of the immediate implication that any trajectory that crosses the $A=1-2M(\bar{\sigma})$ plane Lipschitz continuously represents a solution that can be matched to the TOV$(\bar{\sigma})$ solution. Furthermore, if the trajectory crosses the $A=1-2M(\bar{\sigma})$ plane and intersects the Rankin--Hugoniot curve, which lies in this plane, then the solution can be matched to the TOV$(\bar{\sigma})$ solution to form a general relativistic shock wave. In the case of FLRW$(\sigma,a)$ trajectories, changing the parameters $\sigma$ and $a$ changes the trajectory, so certain combinations of $\sigma$ and $a$ result in an intersection with the Rankine--Hugoniot curve, and thus the formation of an FLRW$(\sigma,a)$--TOV$(\bar{\sigma})$ shock wave. We already know from Theorem \ref{T3} that for $a=1$ the relationship between $\sigma$ and $\bar{\sigma}$ obeys $\bar{\sigma}=H(\sigma)$. For $a\neq1$, trajectories can be generated numerically and the parameters $a$, $\sigma$ and $\bar{\sigma}$ can be adjusted to achieve the intersection. Since the intersection imposes a single constraint on the parameters $a$, $\sigma$ and $\bar{\sigma}$, we conclude that the family of FLRW$(\sigma,a)$--TOV$(\bar{\sigma})$ shock waves is a one-parameter family for each $\sigma$. Fixing $\sigma=\frac{1}{3}$, the resulting family partially answers a claim given in \cite{CT1971} by determining a subset of the pure radiation similarity spacetimes that can be matched to TOV$(\bar{\sigma})$ to form a general relativistic shock wave.

A physically significant shock wave is one for which the equation of state each side of the shock models pure radiation, since these shock waves may have been present during the Radiation Dominated Epoch of the Early Universe. As demonstrated in Figure \ref{F2} for $\sigma=\bar{\sigma}=\frac{1}{3}$, the value of $a$ can be varied in order to achieve an intersection with the Rankine--Hugoniot curve.
\begin{figure}[h]
	\begin{center}
		\includegraphics[width=10cm]{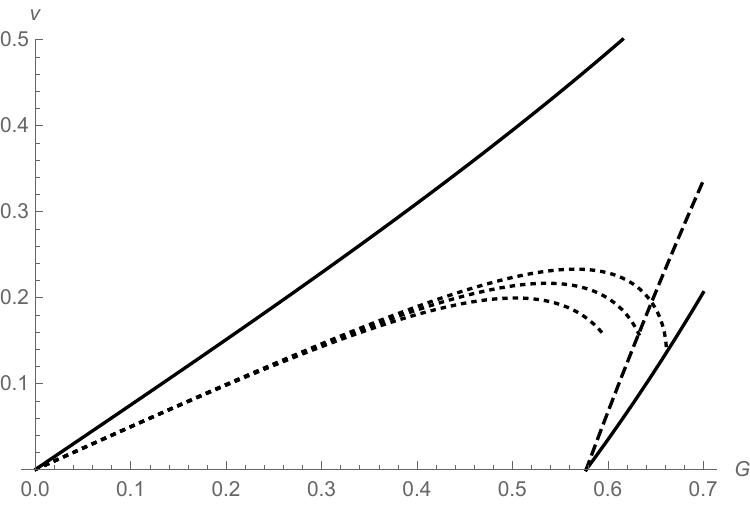}
	\end{center}
	\caption{This figure depicts the same features as Figure \ref{F1}, except the FLRW$(\frac{1}{3},1)$ trajectory is replaced by three FLRW$(\frac{1}{3},a)$ trajectories with varying values of $a$. Unlike in Figure \ref{F1}, the trajectories given in this figure terminate once they reach the $A=1-2M(\bar{\sigma})$ plane.}
	\label{F2}
\end{figure}
In Figure \ref{F2}, the leftmost trajectory overshoots the curve and rightmost trajectory undershoots it. The leftmost, centre and the rightmost trajectories are generated for $a=2.8$, $a=2.58$ and $a=2.4$ respectively. Therefore, the value of the acceleration parameter for the pure radiation shock wave is approximated by
\begin{align*}
	a \approx 2.58,
\end{align*}
with the corresponding point of intersection approximated by
\begin{align*}
	\xi_0 \approx 0.706.
\end{align*}
We recall from Subsection \ref{Subsection2.4} that the FLRW$(\sigma,a)$ spacetimes can exhibit an accelerated expansion that mimics the accelerated expansion found in the Standard Model of Cosmology when close to the radial centre. It is conjectured by Temple that the accelerated expansion observed today is not the result of dark energy, but instead the result of being within a vast primordial shock wave with an FLRW$(\sigma,a)$ interior. What is meant by vast is a shock wave with a shock surface that lies beyond the Hubble radius, that is, not presently observable. What makes this proposal particularly interesting is that the magnitude of acceleration, parameterised by $a$, is determined purely mathematically by the equation of state parameter each side of the shock, assuming a TOV$(\bar{\sigma})$ exterior. However, with $a\approx2.58$, the pure radiation shock wave exhibits an accelerated expansion many orders of magnitude larger than what is observed today. Furthermore, the pure radiation shock surface lies within the Hubble radius. Each of these properties rule out the FLRW$(\frac{1}{3},a)$--TOV$(\frac{1}{3})$ shock wave as a cosmological model, but does not rule out a shock-wave cosmological model consisting of an interior FLRW$(\sigma,a)$ spacetime matched to a non-TOV$(\bar{\sigma})$ exterior. For $a=1$, one such shock wave is constructed with the shock surface lying beyond the Hubble radius by Smoller and Temple in \cite{ST2003}, with the full details of this construction provided in \cite{ST2004}. It remains an open problem to construct shock waves with FLRW$(\sigma,a)$ interiors for which $a\neq1$ and the resulting shock surface lies beyond the Hubble radius.

\section{Shock Waves with TOV Exteriors}\label{Section4}

\subsection{The Lax Characteristic Conditions for TOV Exteriors}\label{Subsection4.1}

As remarked in Subsection \ref{Subsection2.7}, the Lax characteristic conditions provide a notion of stability for general relativistic shock waves and provide an entropy argument for determining whether the shock wave is expected to expand or contract. The Lax characteristic conditions considered in this section apply to expanding shock waves, meaning that shock waves that satisfy these conditions are expected to remain expanding shock waves.

\begin{lemma}
	\label{L3}
	Let $(A,G,v)$ denote a similarity solution to the Einstein--Euler equations with equation of state $p=\sigma\rho$. If there exists a $\xi_0>0$ such that $(A,G,v)$ can be matched to TOV$(\bar{\sigma})$ to form a shock-wave solution, then the Lax characteristic conditions are given by
	\begin{align}
		\frac{\sqrt{\bar{\sigma}}-v(\xi_0)}{1-\sqrt{\bar{\sigma}}v(\xi_0)} < \frac{G(\xi_0)-v(\xi_0)}{1-G(\xi_0)v(\xi_0)} < \sqrt{\sigma}.\label{4.1}
	\end{align}
\end{lemma}
\begin{proof}
	As a reverse to the coordinate transformation introduced in the proof of Proposition \ref{P3}, we begin by transforming a general solution given in self-similar Schwarzchild coordinates, to a solution given in self-similar comoving coordinates. Noting that $B$ and $\vec{u}$ are given implicitly by the triple $(A,G,v)$, we can write this solution in self-similar Schwarzschild coordinates as so:
	\begin{align*}
		ds^2 &= -B(\xi)dt^2 + \frac{1}{A(\xi)}dr^2 + r^2d\Omega^2,\\
		\vec{u} &= (u^0,u^1,0,0),
	\end{align*}
	where $p$ and $\rho$ are determined by (\ref{2.4}) and (\ref{2.8}) respectively. In self-similar comoving coordinates, the solution can be written as:
	\begin{align*}
		d\hat{s}^2 &= -e^{2\varphi}d\hat{t}^2 + e^{2\psi}d\hat{r}^2 + \mathscr{R}^2\hat{r}^2d\Omega^2,\\
		\vec{u} &= (\hat{u}^0,0,0,0).
	\end{align*}
	To eliminate the radial component of the four-velocity, the transformation from Schwarzschild to comoving coordinates must satisfy
	\begin{align*}
		\hat{u}^1 = u^0\frac{\partial\hat{r}}{\partial t} + u^1\frac{\partial\hat{r}}{\partial r} = 0,
	\end{align*}
	which is equivalent to
	\begin{align}
		\frac{\partial\hat{r}}{\partial t} = -\frac{\xi v}{G}\frac{\partial\hat{r}}{\partial r}.\label{4.2}
	\end{align}
	Now given that:
	\begin{align*}
		d\hat{t} &= \frac{\partial\hat{t}}{\partial t}dt + \frac{\partial\hat{t}}{\partial r}dr,\\
		d\hat{r} &= \frac{\partial\hat{r}}{\partial t}dt + \frac{\partial\hat{r}}{\partial r}dr,
	\end{align*}
	then:
	\begin{align*}
		dt &= \left(\frac{\partial\hat{t}}{\partial t}\frac{\partial\hat{r}}{\partial r}-\frac{\partial\hat{t}}{\partial r}\frac{\partial\hat{r}}{\partial t}\right)^{-1}\left(\frac{\partial\hat{r}}{\partial r}d\hat{t}-\frac{\partial\hat{t}}{\partial r}d\hat{r}\right),\\
		dr &= \left(\frac{\partial\hat{t}}{\partial t}\frac{\partial\hat{r}}{\partial r}-\frac{\partial\hat{t}}{\partial r}\frac{\partial\hat{r}}{\partial t}\right)^{-1}\left(-\frac{\partial\hat{r}}{\partial t}d\hat{t}+\frac{\partial\hat{t}}{\partial t}d\hat{r}\right).
	\end{align*}
	Thus to keep the metric diagonal, the condition
	\begin{align*}
		B\frac{\partial\hat{r}}{\partial r}\frac{\partial\hat{t}}{\partial r} - \frac{1}{A}\frac{\partial\hat{r}}{\partial t}\frac{\partial\hat{t}}{\partial t} = 0
	\end{align*}
	 is also needed, which by (\ref{4.2}) is equivalent to
	\begin{align}
		\frac{\partial\hat{t}}{\partial r} = -\frac{Gv}{\xi}\frac{\partial\hat{t}}{\partial t}.\label{4.3}
	\end{align}
	The most general transformation that preserves self-similarity takes the form:
	\begin{align*}
		\hat{t} &= \mathcal{T}(\xi)t,\\
		\hat{r} &= \mathcal{R}(\xi)r,
	\end{align*}
	and conditions (\ref{4.2}) and (\ref{4.3}) determine the functions $\mathcal{T}(\xi)$ and $\mathcal{R}(\xi)$. In self-similar Schwarzschild coordinates the shock speed is given by $\xi=\xi_0$, so in self-similar comoving coordinates the shock speed is given by $\hat{\xi}=\hat{\xi}_0$, where
	\begin{align*}
		\hat{\xi} = \frac{\hat{r}}{\hat{t}} = \frac{\mathcal{R}(\xi)r}{\mathcal{T}(\xi)t} = \frac{\mathcal{R}(\xi)}{\mathcal{T}(\xi)}\xi.
	\end{align*}
	Thus by Lemma \ref{L1} the shock speed is given in interior locally Minkowskian coordinates by
	\begin{align*}
		e^{\psi-\varphi}\hat{\xi}_0.
	\end{align*}
	By Proposition \ref{P8}, it remains to determine $e^{\psi-\varphi}$ and $\tilde{w}$. In this light
	\begin{align*}
		e^{2\varphi} &= \frac{1}{A}\left(\frac{\partial\hat{r}}{\partial t}\right)^2 - \frac{\xi^2}{AG^2}\left(\frac{\partial\hat{r}}{\partial r}\right)^2\\
		&= \frac{\xi^2(1-v^2)}{AG^2}\left(\frac{\partial\hat{r}}{\partial r}\right)^2
	\end{align*}
	and
	\begin{align*}
		e^{2\psi} &= \frac{1}{A}\left(\frac{\partial\hat{t}}{\partial t}\right)^2 - \frac{\xi^2}{AG^2}\left(\frac{\partial\hat{t}}{\partial r}\right)^2\\
		&= \frac{1-v^2}{A}\left(\frac{\partial\hat{t}}{\partial t}\right)^2.
	\end{align*}
	Now:
	\begin{align*}
		\frac{\partial\hat{r}}{\partial t} &= -\xi^2\mathcal{R}'(\xi),\\
		\frac{\partial\hat{r}}{\partial r} &= \mathcal{R}(\xi) + \xi\mathcal{R}'(\xi),
	\end{align*}
	so (\ref{4.2}) gives
	\begin{align*}
		-\xi^2\mathcal{R}' = -\frac{\xi v}{G}\big(\mathcal{R}+\xi\mathcal{R}'\big),
	\end{align*}
	or equivalently
	\begin{align*}
		\xi\mathcal{R}' = \frac{v}{G-v}\mathcal{R}.
	\end{align*}
	Similarly:
	\begin{align*}
		\frac{\partial\hat{t}}{\partial t} &= \mathcal{T}(\xi) - \xi\mathcal{T}'(\xi),\\
		\frac{\partial\hat{t}}{\partial r} &= \mathcal{T}'(\xi),
	\end{align*}
	to which (\ref{4.3}) yields
	\begin{align*}
		\mathcal{T}' = -\frac{Gv}{\xi}\big(\mathcal{T}-\xi\mathcal{T}'\big)
	\end{align*}
	or equivalently
	\begin{align*}
		\xi\mathcal{T}' = -\frac{Gv}{1-Gv}\mathcal{T}.
	\end{align*}
	Therefore the shock speed is given in interior locally Minkowskian coordinates by
	\begin{align*}
		e^{\psi-\varphi}\hat{\xi}_0
		&= G(\xi_0)\frac{\partial\hat{t}}{\partial t}\left(\frac{\partial\hat{r}}{\partial r}\right)^{-1}\frac{\mathcal{R}(\xi_0)}{\mathcal{T}(\xi_0)}\\
		&= \frac{G(\xi_0)-v(\xi_0)}{1-G(\xi_0)v(\xi_0)}.
	\end{align*}
	By Proposition \ref{P2}, TOV$(\bar{\sigma})$ is comoving in Schwarzschild coordinates, and given that TOV$(\bar{\sigma})$ is matched to $(A,G,v)$ in $(t,r)$ coordinates, then the $(\bar{t},\bar{r})$ coordinates of Proposition \ref{P8} are identified with $(t,r)$, so
	\begin{align*}
		\tilde{w} &= e^{\psi-\varphi}\frac{\partial\hat{r}}{\partial\bar{t}}\left(\frac{\partial\hat{t}}{\partial\bar{t}}\right)^{-1}\\
		&= e^{\psi-\varphi}\frac{\partial\hat{r}}{\partial t}\left(\frac{\partial\hat{t}}{\partial t}\right)^{-1}\\
		&= \frac{G(\xi_0)}{\xi_0}\frac{\partial\hat{r}}{\partial t}\left(\frac{\partial\hat{r}}{\partial r}\right)^{-1}\\
		&= -v(\xi_0).
	\end{align*}
	Finally, substituting $e^{\psi-\varphi}\hat{\xi}_0$, $\tilde{w}$ and the equations of state into (\ref{2.40}) yields (\ref{4.1}).
\end{proof}
The following theorem was first proved in \cite{ST1995}, but can instead be obtained directly from Lemma \ref{L3}. For the theorem given in \cite{ST1995}, the value of $\sigma_1$ was only given approximately, however the new theorem provides an exact value.
\begin{theorem}
	\label{T4}
	The expanding FLRW$(\sigma,1)$--TOV$(\bar{\sigma})$ shock-wave solutions satisfy the Lax characteristic conditions for $0<\sigma<\sigma_1$, where
	\begin{align*}
		\sigma_1 = \frac{1+\sqrt{10}}{9} \approx 0.462.
	\end{align*}
\end{theorem}
\begin{proof}
	By Lemma \ref{L2} and Proposition \ref{P4} we know that FLRW$(\sigma,1)$ satisfies (\ref{3.2}) and
	\begin{align}
		G(\xi_0) = \frac{1}{2}(3+3\sigma)v(\xi_0)\left(1+\frac{1}{2}(1+3\sigma)v^2(\xi_0)\right)^{-1}\label{4.4}
	\end{align}
	at the point of intersection with the shock surface. Solving (\ref{3.2}) and (\ref{4.4}) for $G(\xi_0)$ and $v(\xi_0)$ thus gives:
	\begin{align}
		G(\xi_0) &= \frac{1}{2}(3+\bar{\sigma})v(\xi_0),\label{4.5}\\
		v(\xi_0) &= \sqrt{\frac{2(3\sigma-\bar{\sigma})}{(1+3\sigma)(3+\bar{\sigma})}}.\label{4.6}
	\end{align}
	Therefore by using (\ref{3.4}), (\ref{4.5}) and (\ref{4.6}), the left hand inequality of (\ref{4.1}) is found to be satisfied for $0<\sigma<1$ and the right hand inequality is found to be satisfied for for $0<\sigma<\sigma_1$.
\end{proof}

\subsection{The Lax Characteristic Conditions for Subluminal Shocks}\label{Subsection4.2}

\begin{lemma}
	\label{L4}
	Let $(A,G,v)$ denote a similarity solution to the Einstein--Euler equations with equation of state $p=\sigma\rho$. If there exists a $\xi_0>0$ such that $(A,G,v)$ can be matched to TOV$(\bar{\sigma})$ to form a shock-wave solution, then the shock speed is subluminal if
	\begin{align}
		G(\xi_0) < 1\label{4.7}
	\end{align}
	and in such a case the Lax characteristic conditions reduce to:
	\begin{align}
		G(\xi_0) &> \sqrt{\bar{\sigma}},\label{4.8}\\
		\{\cdot\}_D(\xi_0) &< 0.\label{4.9}
	\end{align}
\end{lemma}
\begin{proof}
	By Lemma \ref{L3}, the shock speed is subluminal if
	\begin{align*}
		\frac{G(\xi_0)-v(\xi_0)}{1-G(\xi_0)v(\xi_0)} < 1,
	\end{align*}
	which for $0<v<1$ is equivalent to (\ref{4.7}). For $G<1$ it is then not difficult to check the left hand inequality of (\ref{4.1}) is equivalent to (\ref{4.8}). It thus remains to demonstrate the right hand inequality is equivalent to (\ref{4.9}). In this light, we have
	\begin{align*}
		\{\cdot\}_D 
		&= \frac{3}{4}(3+3\sigma)\left[(G-v)^2-\sigma(1-Gv)^2\right]\\
		&= \frac{3}{4}(3+3\sigma)\left[G-v+\sqrt{\sigma}(1-Gv)\right]\left[G-v-\sqrt{\sigma}(1-Gv)\right],
	\end{align*}
	and for $0<v<G<1$ we see that $\{\cdot\}_D=0$ is equivalent to
	\begin{align*}
		\frac{G-v}{1-Gv} = \sqrt{\sigma},
	\end{align*}
	which completes the proof.
\end{proof}
The following theorem, first proved in \cite{ST1995}, demonstrates that stable FLRW$(\sigma,1)$--TOV$(\bar{\sigma})$ shock waves can be constructed with shock speeds arbitrarily close to the speed of light.
\begin{theorem}
	\label{T5}
	The FLRW$(\sigma,1)$--TOV$(\bar{\sigma})$ shock-wave solutions have subluminal shock speeds for $0<\sigma<\sigma_2$, where
	\begin{align*}
		\sigma_2 = \frac{\sqrt{5}}{3} \approx 0.745.
	\end{align*}
\end{theorem}
\begin{proof}
	This follows directly from Lemma \ref{L4} and relations (\ref{3.4}), (\ref{4.5}) and (\ref{4.6}).
\end{proof}
Before introducing Theorem \ref{T6}, a few definitions are in order.
\begin{definition}
	\label{D11}
	The \emph{singular surface} and \emph{sonic surface} are defined in $(A,G,v)$ space by $\{\cdot\}_S=0$ and $\{\cdot\}_D=0$ respectively. Moreover, the \emph{subsonic region} and \emph{supersonic region} are defined by $\{\cdot\}_D<0$ and $\{\cdot\}_D>0$ respectively. Furthermore, solutions whose trajectories remain in the subsonic region are referred to as \emph{subsonic}, those that remain in the supersonic region as \emph{supersonic} and those that pass through the sonic surface as \emph{transonic}. A point where a trajectory intersects the sonic surface is referred to as a \emph{sonic point} for that trajectory.
\end{definition}
As a consequence of Lemma \ref{L4}, the sonic surface serves as a convenient indicator for satisfying one of the Lax characteristic conditions of an expanding general relativistic shock wave with a TOV$(\bar{\sigma})$ exterior. This is particularly useful for numerical approximations, since if the intersection with the Rankine--Hugoniot curve lies in the subsonic region and to the right of the $G=\sqrt{\bar{\sigma}}$ plane, then the Lax characteristic conditions are satisfied. Note that the subsonic region is everything left of the right most unbroken curve in Figure \ref{F3}. For $\sigma\neq\bar{\sigma}$, condition (\ref{4.8}) is not automatically satisfied, since the Rankine--Hugoniot curve does not intersect the $v=0$ plane at $G=\sqrt{\sigma}$, as Figure \ref{F3} demonstrates.
\begin{figure}[h]
	\begin{center}
		\includegraphics[width=10cm]{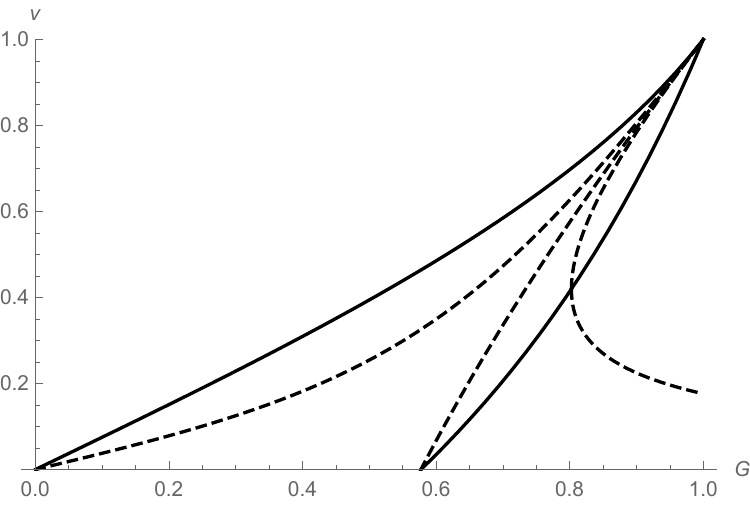}
	\end{center}
	\caption{This figure depicts the singular surface as the left unbroken curve, the sonic surface as the right unbroken curve and three Rankine--Hugoniot curves by dashed curves, all for $\sigma=\frac{1}{3}$.}
	\label{F3}
\end{figure}

In Figure \ref{F3}, the leftmost, centre and rightmost dashed curves correspond to:
\begin{align*}
	\sigma &< \bar{\sigma} = \frac{2}{3}, & \sigma &= \bar{\sigma} = \frac{1}{3}, & \sigma &> \bar{\sigma} = \frac{1}{6},
\end{align*}
respectively. In the $\sigma<\bar{\sigma}$ case, the Rankine--Hugoniot curve touches the singular surface at $(G,v)=(0,0)$ and $(G,v)=(1,1)$. In the $\sigma=\bar{\sigma}$ case, the Rankine--Hugoniot curve touches the sonic surface at $(G,v)=(\sqrt{\sigma},0)$ and $(G,v)=(1,1)$. In the $\sigma>\bar{\sigma}$ case, the Rankine--Hugoniot curve touches the sonic surface at $(G,v)=(1,1)$ and intersects it at:
\begin{align*}
	G &= \frac{\sqrt{\sigma}(1+\bar{\sigma})+\sqrt{(\sigma-\bar{\sigma})(1-\sigma\bar{\sigma})}}{1+\sigma}, & v &= \sqrt{\frac{\sigma-\bar{\sigma}}{1-\sigma\bar{\sigma}}}.
\end{align*}
Thus for an expanding shock wave to not satisfy the Lax characteristic conditions, the solution trajectory must either intersect the Rankine--Hugoniot curve before the $G=\sqrt{\bar{\sigma}}$ plane or after passing through the sonic surface. Since conditions $(\ref{4.8})$ and $(\ref{4.9})$ are automatically satisfied for $\sigma=\bar{\sigma}$, then the expanding shock waves for which $\sigma=\bar{\sigma}$ automatically satisfy the Lax characteristic conditions, as the following theorem summarises.
\begin{theorem}
	\label{T6}
	Let $(A,G,v)$ denote a similarity solution to the Einstein--Euler equations with equation of state $p=\sigma\rho$. If there exists a $\xi_0>0$ such that $(A,G,v)$ can be matched to TOV$(\bar{\sigma})$ to form a shock-wave solution with a subluminal shock speed, then the Lax characteristic conditions are satisfied if:
	\begin{enumerate}
		\item $\sigma=\bar{\sigma}$ or
		\item $\sigma<\bar{\sigma}$ and $G(\xi_0)>\sqrt{\bar{\sigma}}$ or
		\item $\sigma>\bar{\sigma}$ and $\{\cdot\}_D(\xi_0)<0$.
	\end{enumerate}
\end{theorem}
\begin{proof}
	This is an immediate consequence of Lemma \ref{L4} and the discussion proceeding Definition \ref{D11}.
\end{proof}

\section{Existence of The Pure Radiation FLRW--TOV Shock Wave}\label{Section5}

\subsection{Monotonicity Lemma}\label{Subsection5.1}

Since Proposition \ref{P4} gives FLRW$(\sigma,1)$ as an implicit function of $v$, we know if and where the solution trajectory crosses the singular and sonic surfaces. However, even though we can expect similar behaviour of FLRW$(\sigma,a)$ for $\xi\ll1$ and $a\approx1$, there is no guarantee that this persists as $\xi$ increases or for larger values of $|a-1|$. In particular, all FLRW$(\sigma,1)$ and TOV$(\sigma)$ solutions are transonic, whereas this is not the case for generic FLRW$(\sigma,a)$ solutions, most of which are only subsonic. All FLRW$(\sigma,a)$--TOV$(\sigma)$ shock waves are transonic, since either the interior FLRW$(\sigma,a)$ or exterior TOV$(\sigma)$ spacetime must pass through the sonic surface to be matched. We know from Figure \ref{F2} that the FLRW$(\frac{1}{3},2.4)$ trajectory differs significantly from the FLRW$(\frac{1}{3},1)$ trajectory as $\xi$ increases, since it encounters a singularity in equation (\ref{2.7}) upon hitting the sonic surface. The following lemma assists in predicting the behaviour of FLRW$(\sigma,a)$ trajectories, as it implies any FLRW$(\sigma,a)$ trajectory will either hit the sonic surface or pass through the plane $A=1-2M(\sigma)$. Recall from Figure \ref{F3} that $\{\cdot\}_D<0$ implies the trajectory is to the left of the sonic surface and $\{\cdot\}_S>0$ implies that the trajectory is below the singular surface. The monotonicity of $A$ and $G$ implies that the trajectory advances to the right whilst simultaneously approaching the $A=1-2M(\sigma)$ plane.
\begin{lemma}
	\label{L5}
	Let $0<\sigma<1$, $a>0$ and $\xi>0$. So long as FLRW$(\sigma,a)$ satisfies:
	\begin{align*}
		A &> 1- 2M(\sigma),\\
		\{\cdot\}_D &< 0,
	\end{align*}
	then it also satisfies:
	\begin{align}
		A' &< 0,\label{5.1}\\
		G' &> 0,\label{5.2}\\
		v &> 0,\label{5.3}\\
		\{\cdot\}_S &> 0.\label{5.4}
	\end{align}
\end{lemma}
\begin{proof}
	From Definition \ref{D6}, for sufficiently small $\xi>0$ we have the inequalities:
	\begin{align*}
		1-A &> 0, & A' &< 0,\\
		G &> 0, & G' &> 0,\\
		v &> 0, & \{\cdot\}_S &> 0,
	\end{align*}
	and so the FLRW$(\sigma,a)$ trajectory begins by satisfying inequalities (\ref{5.1})--(\ref{5.4}). It is thus sufficient to show that each one of the four inequalities is implied by the other three as $\xi$ increases. In this light, assume $v>0$ and $\{\cdot\}_S>0$, then equation (\ref{2.5}) and the small $\xi$ inequality $1-A>0$ implies inequality (\ref{5.1}). For inequality (\ref{5.2}), assume $v>0$ and $\{\cdot\}_S>0$ and note that $\{\cdot\}_S>0$ and $\{\cdot\}_D<0$ imply $v<1$. Given these constraints, equation (\ref{2.6}) implies
	\begin{align*}
		\xi\frac{dG}{d\xi} &= -G\left[\left(\frac{1-A}{A}\right)\frac{(3+3\sigma)[(1+v^2)G-2v]}{2\{\cdot\}_S}-1\right]\\
		&= G\left[1-\left(\frac{1-A}{A}\right)\frac{(3+3\sigma)(1+v^2)G-(6+6\sigma)v}{(6+6\sigma v^2)G-(6+6\sigma)v}\right]\\
		&> G\left[1-\left(\frac{1-A}{A}\right)\right]\\
		&> 0,
	\end{align*}
	with the last line following from (\ref{3.5}) and the small $\xi$ inequality $G'>0$. Now it is sufficient to demonstrate inequality (\ref{5.3}) on the interval $0<G<\sqrt{\sigma}$, since the sonic surface intersects the $v=0$ plane at $G=\sqrt{\sigma}$ and we are assuming that the trajectory stays off the sonic surface. In this light, assume $A'<0$, $G'>0$ and $\{\cdot\}_S>0$. Note that $A'<0$ implies $A<1$, and $G'>0$ implies $G>0$. By equation (\ref{2.7}), the sign of $v'$ on the $v=0$ plane in the region bounded by $1-2M(\sigma)<A<1$ and $0<G<\sqrt{\sigma}$ is strictly positive, since
	\begin{align*}
		\xi\frac{dv}{d\xi}\bigg|_{v=0} &= -\left(\frac{1-v^2}{2\{\cdot\}_D}\right)\left[3\sigma\{\cdot\}_S+\left(\frac{1-A}{A}\right)\frac{(3+3\sigma)^2\{\cdot\}_N}{4\{\cdot\}_S}\right]\bigg|_{v=0}\\
		&= \left(\frac{2G}{3(3+3\sigma)(\sigma-G^2)}\right)\left[9\sigma-\frac{(3+3\sigma)^2}{4}\left(\frac{1-A}{A}\right)\right]\\
		&> 0.
	\end{align*}
	Thus any trajectory that begins above the $v=0$ plane remains above the $v=0$ plane. The small $\xi$ inequality $v>0$ along with this result then implies inequality (\ref{5.3}). Note that such a result still holds when $1-2M(\bar{\sigma})<A<1$ for $0<\bar{\sigma}\leq\sigma<1$, since
	\begin{align*}
		9\sigma - \frac{(3+3\sigma)^2}{4}\left(\frac{1-A}{A}\right) &> 9\sigma - \frac{(3+3\sigma)^2}{4}\left(\frac{2M(\bar{\sigma})}{1-2M(\bar{\sigma})}\right)\\
		&= 9\sigma - \frac{9\bar{\sigma}(3+3\sigma)^2}{(3+3\bar{\sigma})^2}\\
		&= (3+3\sigma)^2\left(\frac{9\sigma}{(3+3\sigma)^2}-\frac{9\bar{\sigma}}{(3+3\bar{\sigma})^2}\right)\\
		&\geq 0.
	\end{align*}
	Finally, inequality (\ref{5.4}) is demonstrated in a similar manner to inequality (\ref{5.3}), that is, by showing that trajectories stay away from the surface $\{\cdot\}_S=mv$ for some $0<m<\frac{3}{2}$. The upper bound for $m$ ensures FLRW$(\sigma,a)$ trajectories initially satisfy inequality (\ref{5.4}). Now assume inequalities (\ref{5.1})--(\ref{5.3}) hold and note that these inequalities additionally imply $A<1$ and $G>0$. Since $\{\cdot\}_S=mv$ is equivalent to
	\begin{align}
		G = \frac{(3+3\sigma+m)v}{3+3\sigma v^2},\label{5.5}
	\end{align}
	then by equation (\ref{2.6}) and (\ref{5.5}), we have
	\begin{align*}
		q_A(v;\sigma,m) &= \xi\frac{d}{d\xi}\left(G-\frac{(3+3\sigma+m)v}{3+3\sigma v^2}\right)\bigg\vert_{\{\cdot\}_S=mv}\\
		&= \left(\xi\frac{dG}{d\xi}-\frac{(3+3\sigma+m)(3-3\sigma v^2)}{(3+3\sigma v^2)^2}\xi\frac{dv}{d\xi}\right)\bigg\vert_{\{\cdot\}_S=mv}\\
		&= \frac{(3+3\sigma+m)v}{3+3\sigma v^2}\left[1-\left(\frac{1-A}{A}\right)\frac{(3+3\sigma)[(3+3\sigma+m)(1+v^2)-2(3+3\sigma v^2)]}{2m(3+3\sigma v^2)}\right]\\
		&+ \frac{(3+3\sigma+m)(3-3\sigma v^2)}{(3+3\sigma v^2)^2}\left(\frac{1-v^2}{2\{\cdot\}_D}\right)\left[3\sigma mv+\left(\frac{1-A}{A}\right)\frac{(3+3\sigma)^2\{\cdot\}_N}{4mv}\right]\\
		&= \frac{(3+3\sigma+m)v}{(3+3\sigma v^2)^2}\left[3+3\sigma v^2+\frac{3\sigma m(1-v^2)(3-3\sigma v^2)}{2\{\cdot\}_D}+\left(\frac{1-A}{A}\right)(\{\cdot\}_A+\{\cdot\}_B+\{\cdot\}_C)\right],
	\end{align*}
	where:
	\begin{align*}
		\{\cdot\}_A &= \frac{(3+3\sigma)(3-3\sigma-m)(1-v^2)}{2m},\\
		\{\cdot\}_B &= \frac{(3+3\sigma)^2(1-v^2)(3-3\sigma v^2)\{\cdot\}_N}{8mv^2\{\cdot\}_D},\\
		\{\cdot\}_C &= -(3+3\sigma)v^2.
	\end{align*}
	Let $0<v_*<1$. The objective for this part is to find an $m$ such that $q_A(v;\sigma,m)>0$ for all $0<\sigma<1$ and $0<v<v_*$. Note that it is always possible to choose an $m$ small enough to ensure $v_*<v_I(\sigma,m)$, where $v_I(\sigma,m)$ is the intersection of surfaces (\ref{5.5}) and $\{\cdot\}_D=0$, since
	\begin{align*}
		\lim\limits_{m\to 0}v_I(\sigma,m) = 1.
	\end{align*}
	Now even though it can be shown that $\{\cdot\}_A+\{\cdot\}_B+\{\cdot\}_C>0$ for a certain interval of $v$, it is easier to use $\{\cdot\}_A+\{\cdot\}_B>0$ for the whole interval $0<v<v_I$. This is the case since
	\begin{align*}
		\{\cdot\}_A + \{\cdot\}_B &= \frac{(3+3\sigma)(1-v^2)}{2m}\left[3-3\sigma-m+\frac{(3+3\sigma)(3-3\sigma v^2)\{\cdot\}_N}{4v^2\{\cdot\}_D}\right]\\
		&= \frac{(3+3\sigma)(1-v^2)}{8m(-\{\cdot\}_D)v^2}\left[4(3-3\sigma-m)(-\{\cdot\}_D)v^2-(3+3\sigma)(3-3\sigma v^2)\{\cdot\}_N\right]\\
		&= \frac{3\sigma(3+3\sigma)^2(1-v^2)(3-3v^2+n)}{8(-\{\cdot\}_D)(3+3\sigma v^2)^2}\left[3-3v^2+\sigma(9+n)v^2-\sigma(9+n\sigma)v^4\right]\\
		&> 0,
	\end{align*}
	where $m=n\sigma$ for some $0<n<\frac{3}{2}$. Given that $\{\cdot\}_A+\{\cdot\}_B>0$ and $\{\cdot\}_C<0$, then for $\frac{1}{2}<A<1$ we have
	\begin{align*}
		\left(\frac{1-A}{A}\right)(\{\cdot\}_A+\{\cdot\}_B+\{\cdot\}_C) > \left(\frac{1-A}{A}\right)\{\cdot\}_C > \{\cdot\}_C
	\end{align*}
	Thus for any $0<\sigma<1$ and $0<v<v_*$,
	\begin{align*}
		\lim\limits_{m\to 0}q_A(v;\sigma,m) &= \lim\limits_{n\to 0}q_A(v;\sigma,n\sigma)\\
		&> \lim\limits_{n\to 0}\frac{(3+3\sigma+n\sigma)v}{(3+3\sigma v^2)^2}\left[3+3\sigma v^2+\frac{3n\sigma^2(1-v^2)(3-3\sigma v^2)}{2\{\cdot\}_D}+\{\cdot\}_C\right]\\
		&= \lim\limits_{n\to 0}\frac{(3+3\sigma+n\sigma)(1-v^2)v}{(3+3\sigma v^2)^2}\left[3-\frac{2n\sigma(3+3\sigma)^{-1}(3-3\sigma v^2)(3+3\sigma v^2)^2}{(3-3v^2-n\sigma v^2)^2-\sigma v^2(3-3v^2+n)^2}\right]\\
		&> 0.
	\end{align*}
	Therefore, for any interval $0<v<v_*$ with $v_*<1$, there exists an $0<n<\frac{3}{2}$ such that the surface $\{\cdot\}_S=n\sigma v$ cannot be crossed. Now assume for contradiction that a trajectory crosses the $\{\cdot\}_S=0$ surface. Because $v_I(\sigma,0)=1$ and we assume $\{\cdot\}_D<0$, the trajectory cannot cross the surface $\{\cdot\}_S=0$ at $v=1$, so it must intersect at some point $0<v_{**}< 1$. Given that FLRW$(\sigma,a)$ satisfies $\{\cdot\}_S>n\sigma v$ initially for any $0<n<\frac{3}{2}$ and we can pick a $v_*$ such that $v_{**}<v_*<1$, we know that the surface $\{\cdot\}_S=n\sigma v$ cannot be crossed on the interval $0<v<v_*$, which is a contradiction. Thus under our assumptions, FLRW$(\sigma,a)$ satisfies inequality (\ref{5.4}) and completes the proof.
\end{proof}

\subsection{Existence Theorem}\label{Subsection5.2}

In Section \ref{Section3}, the pure radiation shock wave is constructed numerically, with the acceleration parameter approximated by $a\approx2.58$. Now with Lemma \ref{L5} in place, we are in a position to provide a rigorous existence proof for this shock wave, which constitutes the main result of this paper.
\begin{theorem}
	\label{T7}
	There exists an $a>1$ such that FLRW$(\frac{1}{3},a)$ can be matched to TOV$(\frac{1}{3})$ to form a pure radiation general relativistic shock wave that satisfies the Lax characteristic conditions.
\end{theorem}

\subsubsection{Outline of Proof}

From Lemma \ref{L5} we know that any FLRW$(\sigma,a)$ trajectory inevitably hits either the sonic surface or the matching surface $A=1-2M(\sigma)$, since by equation (\ref{2.5}) the derivative of $A$ remains bounded above by a value strictly less than zero, that is, the trajectory cannot asymptote to a constant value of $A$. The first part of the proof argues that it is sufficient to find a constant $b>1$ such that the FLRW$(\frac{1}{3},b)$ trajectory overshoots the Rankine--Hugoniot curve when it hits the matching surface. This is sufficient because we know the explicitly given FLRW$(\frac{1}{3},1)$ trajectory undershoots the Rankine--Hugoniot curve when it hits the matching surface, so the continuity of the parameter $a$ guarantees an intersection for some $1<a<b$.

The next part of the proof also follows from Lemma \ref{L5} in the sense that the monotonicity of $G$ permits $A$ and $v$ to be considered as explicit functions of $G$, rather than functions of $\xi$. The remainder of the proof consists of the computationally difficult task of building a trapping region around the FLRW$(\frac{1}{3},b)$ trajectory to guarantee that it overshoots the Rankine--Hugoniot curve.

The trapping region is constructed from a perturbation of high-order Taylor polynomials of $A$ and $v$. These perturbed polynomials are functions of $G$, and bound $A$ and $v$ from above and below all the way to the matching surface. There are multiple parameters involved in constructing such a region, such as the magnitude of perturbation from the Taylor polynomials, the choice of $b$ and the order of the Taylor polynomials. Each choice of parameter affects the choice of the others, so a careful balance must be found between them.

Due to the complexity of equations (\ref{2.5})--(\ref{2.7}), the high-order Taylor polynomials have rational coefficients that are too large to be written in any meaningful way, and doing so would more than double the length of this paper. Instead, a procedure is given for generating these coefficients using symbolic manipulation software, from which they can be analysed and simplified as necessary.

The most difficult part in choosing the parameters, is ensuring the resulting perturbed Taylor region satisfies a set of inequalities that imply vectors on the boundary of the region point into the region, thus trapping any trajectory that begins inside. These choices are made difficult by virtue of the fact that the matching for $\sigma=\bar{\sigma}=\frac{1}{3}$ occurs near the sonic surface, where the Taylor polynomials are approximating a function approaching a singularity.

The final step is demonstrating analytically the choice of perturbed Taylor region is in fact a trapping region containing the FLRW$(\frac{1}{3},b)$ trajectory from $G=0$ to the matching surface. The Lax characteristic conditions then follow by Theorem \ref{T6}.

\subsubsection{The Shooting Argument}

With reference to Figure \ref{F1}, which is a two-dimensional slice of the $(A,G,v)$ phase space, we want the trajectory of FLRW$(\frac{1}{3},b)$ to reach the $A=\frac{4}{7}$ surface before the trajectory intersects the sonic surface. In particular, we want the FLRW$(\frac{1}{3},b)$ trajectory to intersect the $A=\frac{4}{7}$ surface between the Rankine--Hugoniot curve and the sonic surface, which is the region bounded by the dashed line and solid line in Figure \ref{F1}. This corresponds to an overshoot of the Rankine--Hugoniot curve.

For FLRW$(\frac{1}{3},a)$ to match with TOV$(\frac{1}{3})$ to form a general relativistic shock wave, then, by Lemma \ref{L2} and Definition \ref{D10}, there must exist a $\xi_0$ such that FLRW$(\frac{1}{3},a)$ satisfies:
\begin{align}
	A(\xi_0) &= \frac{4}{7},\label{5.6}\\
	v(\xi_0) &= \Gamma_{RH}\left(G(\xi_0);\frac{1}{3},\frac{1}{3}\right).\label{5.7}
\end{align}
We know from Theorem \ref{T3} that FLRW$(\frac{1}{3},1)$ cannot form a general relativistic shock wave with TOV$(\frac{1}{3})$, since $\sigma=\bar{\sigma}=\frac{1}{3}$ is not a solution of $\bar{\sigma}=H(\sigma)$. Instead, the intersection of the FLRW$(\frac{1}{3},1)$ trajectory with the $A=\frac{4}{7}$ surface results in
\begin{align*}
	v(\xi_0) < \Gamma_{RH}\left(G(\xi_0);\frac{1}{3},\frac{1}{3}\right),
\end{align*}
that is, the FLRW$(\frac{1}{3},1)$ trajectory undershoots the Rankine--Hugoniot curve for $\bar{\sigma}=\frac{1}{3}$. Note that the explicitly known FLRW$(\frac{1}{3},1)$ solution is able to cross the sonic surface without becoming singular due to the cancellation of $\{\cdot\}_D$ in equation (\ref{2.7}) at the sonic point, which we recall is the point of intersection with the sonic surface. General FLRW$(\sigma,a)$ solutions typically become singular at the sonic point.

Now suppose that there exists a $b>1$ such that the FLRW$(\frac{1}{3},b)$ trajectory intersects the plane $A=\frac{4}{7}$ with
\begin{align}
	v(\xi_0) > \Gamma_{RH}\left(G(\xi_0);\frac{1}{3},\frac{1}{3}\right),\label{5.8}
\end{align}
then providing the transition of the FLRW$(\frac{1}{3},1)$ trajectory to the FLRW$(\frac{1}{3},b)$ trajectory crosses the Rankine--Hugoniot curve, there exists an $1<a<b$ such that (\ref{5.6}) and (\ref{5.7}) are satisfied. An example of this process is demonstrated numerically in Figure \ref{F2}. Lemma \ref{L5} establishes the fact that if the FLRW$(\frac{1}{3},a)$ trajectory remains in the subsonic region, then it must eventually hit the $A=\frac{4}{7}$ plane. The continuous dependence of FLRW$(\frac{1}{3},a)$ on the parameter $a$ means that there is a continuous transition from FLRW$(\frac{1}{3},1)$ to FLRW$(\frac{1}{3},b)$, at least up until the trajectory hits the $A=\frac{4}{7}$ plane or hits the sonic surface. This continuous transition, along with Lemma \ref{L5}, guarantees the crossing of the Rankine--Hugoniot curve in the $\sigma=\bar{\sigma}=\frac{1}{3}$ case, since the transition from hitting the sonic surface to hitting the $A=\frac{4}{7}$ plane occurs on the intersection of the sonic surface with the $A=\frac{4}{7}$ plane. Thus it is sufficient to rigorously demonstrate the existence of an FLRW$(\frac{1}{3},b)$ solution that satisfies (\ref{5.6}) and (\ref{5.8}). We know from Figure \ref{F2} that a numerical approximation of the FLRW$(\frac{1}{3},2.8)$ trajectory passes above the Rankine--Hugoniot curve, so for this reason, and another reason explained later in the proof, existence is considered for
\begin{align*}
	b = \frac{17}{6} > 2.8.
\end{align*}

\subsubsection{The Monotonicity Simplification}

The FLRW$(\frac{1}{3},a)$ trajectories originate from the unstable fixed point $(A,G,v)=(1,0,0)$, which means distinct trajectories will generically diverge from one another. If a surface is constructed closely enclosing a trajectory, then the vectors generated by system (\ref{2.5})--(\ref{2.7}) on this surface will point inward if the surface expands or loosens faster than the trajectories diverge from one another. Conversely, if the surface remains too tightly bound around a trajectory, then one would expect the vectors on the surface of the trapping region to point outward. These facts guide the construction of a trapping region around the trajectory of FLRW$(\frac{1}{3},\frac{17}{6})$, which is how this trajectory is shown to overshoot the Rankine--Hugoniot curve. In particular, these facts highlight a balancing act in the construction of the trapping region, since on the one hand, a looser trapping region will be easier to show that trajectories remain within it, whilst a tighter trapping region will be easier to show where the trajectory will land on the $A=\frac{4}{7}$ plane.

Since Lemma \ref{L5} establishes the monotonicity of $G$ as a function of $\xi$, $A$ and $v$ can be written as functions of $G$, with equations (\ref{2.5})--(\ref{2.7}) becoming:
\begin{align}
	\frac{dA}{dG} &= -\left(\xi\frac{dG}{d\xi}\right)^{-1}\frac{(3+3\sigma)(1-A)v}{\{\cdot\}_S}\label{5.9},\\
	\xi\frac{dG}{d\xi} &= -G\left[\left(\frac{1-A}{A}\right)\frac{(3+3\sigma)[(1+v^2)G-2v]}{2\{\cdot\}_S}-1\right],\notag\\
	\frac{dv}{dG} &= -\left(\xi\frac{dG}{d\xi}\right)^{-1}\left(\frac{1-v^2}{2\{\cdot\}_D}\right)\left[3\sigma\{\cdot\}_S+\left(\frac{1-A}{A}\right)\frac{(3+3\sigma)^2\{\cdot\}_N}{4\{\cdot\}_S}\right].\label{5.10}
\end{align}
In this light, the trajectory of FLRW$(\frac{1}{3},\frac{17}{6})$ can be represented as $(A(G),v(G))$, with $G$ parameterising the progress of the trajectory towards the $A=\frac{4}{7}$ plane. This step provides a considerable simplification, since any trapping region now only needs to bound $A$ and $v$.

\subsubsection{Generating the Taylor Polynomials}

To generate the Taylor polynomials of $A(G)$ and $v(G)$ from equations (\ref{5.9}) and (\ref{5.10}), we begin with Definition \ref{D6}, which defines FLRW$(\sigma,a)$ to leading order in $\xi$ by:
\begin{align*}
	A(\xi) &= 1 - \frac{1}{4}a^2\xi^2 + O_{\xi\to0}(\xi^4),\\
	G(\xi) &= \frac{1}{4}(3+3\sigma)\xi + O_{\xi\to0}(\xi^3),\\
	v(\xi) &= \frac{1}{2}\xi + O_{\xi\to0}(\xi^3).
\end{align*}
Thus to leading order in $G$, we have:
\begin{align*}
	A(G) &= 1 - \frac{4a^2}{(3+3\sigma)^2}G^2 + O_{G\to0}(G^4),\\
	v(G) &= \frac{2}{3+3\sigma}G + O_{G\to0}(G^3).
\end{align*}
Now to obtain the third-order polynomials, we first substitute:
\begin{align*}
	A(G) &= 1 - \frac{4a^2}{(3+3\sigma)^2}G^2,\\
	v(G) &= \frac{2}{3+3\sigma}G + \frac{v^{(3)}(0)}{3!}G^3,
\end{align*}
into equations (\ref{5.9}) and (\ref{5.10}) and then solve for $v^{(3)}(0)$ to eliminate the leading order term. This gives
\begin{align*}
	\frac{v^{(3)}(0)}{3!} = \frac{2(1-a^2+14\sigma-4a^2\sigma+33\sigma^2-3a^2\sigma^2)}{5\sigma(3+3\sigma)^3}.
\end{align*}
The same procedure can be used to determine $A^{(4)}(0)$, which is given by
\begin{align*}
	\frac{A^{(4)}(0)}{4!} = -\frac{4a^2(3-3a^2+42\sigma-22a^2\sigma+39\sigma^2+21a^2\sigma^2)}{5\sigma(3+3\sigma)^4}.
\end{align*}
Now this process can be continued to determine $v^{(5)}(0)$, then $A^{(6)}(0)$ and so on. Symbolic manipulation software is required for higher order terms, since for $v^{(5)}(0)$ we already find
\begin{multline*}
	-\frac{5\sigma^2(3+3\sigma)^5}{2}\frac{v^{(5)}(0)}{5!} = 1 - 40\sigma - 642\sigma^2 - 2568\sigma^3 - 3087\sigma^4 - a^4 + 48a^2\sigma - 8a^4\sigma + 348a^2\sigma^2\\ + 768a^2\sigma^3 + 468a^2\sigma^4 + 14a^4\sigma^2 + 120a^4\sigma^3 + 99a^4\sigma^4,
\end{multline*}
and by the time we reach $v^{(33)}(0)$ we require almost 30000 characters to write this single coefficient.

\subsubsection{Perturbing the Taylor Polynomials}

The next step is to construct a trapping region around the trajectory components $A$ and $v$. This will be achieved using a modification of the Taylor polynomials of $A$ and $v$ about $G=0$. In this light, define:
\begin{align*}
	A_M(G) &= \sum_{n=0}^{N-1}\frac{A^{(2n)}(0)}{(2n)!}G^{2n} + M_A G^{2N},\\
	A_m(G) &= \sum_{n=0}^{N-1}\frac{A^{(2n)}(0)}{(2n)!}G^{2n} + m_A G^{2N},\\
	v_M(G) &= \sum_{n=0}^{N-1}\frac{v^{(2n+1)}(0)}{(2n+1)!}G^{2n+1} + M_v G^{2N+1},\\
	v_m(G) &= \sum_{n=0}^{N-1}\frac{v^{(2n+1)}(0)}{(2n+1)!}G^{2n+1} + m_v G^{2N+1},
\end{align*}
where $M_A$, $m_A$, $M_v$ and $m_v$ are chosen so that:
\begin{align*}
	m_A &< \frac{A^{(2N)}(0)}{(2N)!} < M_A,\\
	m_v &< \frac{v^{(2N+1)}(0)}{(2N+1)!} < M_v.
\end{align*}
The functions $A_M$ and $A_m$ are used to bound $A$ from above and below respectively, with $v_M$ and $v_m$ providing analogous bounds for $v$. One of the objectives is to show
\begin{align}
	v_m(G_0) > \Gamma_{RH}\left(G_0;\frac{1}{3},\frac{1}{3}\right),\label{5.11}
\end{align}
where $G_0$ is found implicitly through
\begin{align}
	A_M(G_0) = \frac{4}{7}.\label{5.12}
\end{align}
This is so the lowest point of the trapping region bounding $v$ remains above the Rankine--Hugoniot curve at $G_0$. The trapping region bounding $v$ is depicted in Figure \ref{F5}. The value of $G_0$ is determined by the intersection of $A_M$ with the $A=\frac{4}{7}$ plane. The reason $A_M$ is chosen over $A_m$, is because $A_M$ reaches the $A=\frac{4}{7}$ plane later than $A_m$, that is, at a larger value of $G$. So defining $G_0$ as the point in which $A_M=\frac{4}{7}$ provides an upper bound for value of $G$ in which $A=\frac{4}{7}$. The trapping region bounding $A$ is depicted in Figure \ref{F4}.

\subsubsection{Determining the Perturbed Taylor Region}

For large enough $N$, it is possible to find values for $M_A$, $m_A$, $M_v$ and $m_v$ such that (\ref{5.11}) and (\ref{5.12}) are satisfied and inequalities:
\begin{align}
	A_m(G) &< A < A_M(G),\label{5.13}\\
	v_m(G) &< v < v_M(G),\label{5.14}
\end{align}
hold on $0<G<G_0$. Although, as will be seen in the last step, the values are actually selected to satisfy (\ref{5.11}), (\ref{5.12}) and (\ref{5.17})--(\ref{5.20}) instead, which then imply (\ref{5.11})--(\ref{5.14}). The values are selected through extensive trial and error, using numerical approximations of $A$ and $v$ as a guide. Note that the Taylor polynomials of $A$ and $v$ converge quicker for smaller values of $b$, but larger values of $b$ allow for (\ref{5.11}) to be more easily satisfied. This is why $b=\frac{17}{6}$ is chosen, as it provides a sufficient compromise. Through this process it is found that $N=16$ and the following values satisfy (\ref{5.11}), (\ref{5.12}) and (\ref{5.17})--(\ref{5.20}):
\begin{align*}
	M_A &= (1+2^{-6})\frac{A^{(32)}(0)}{(32)!},\\
	m_A &= 2^{-1}\frac{A^{(32)}(0)}{(32)!},\\
	M_v &= 2^{-1}\frac{v^{(33)}(0)}{(33)!},\\
	m_v &= 2^5\frac{v^{(33)}(0)}{(33)!},
\end{align*}
noting that $M_v$ and $m_v$ are chosen in the knowledge that $v^{(33)}(0)$ is negative.

With $M_A$, $m_A$, $M_v$ and $m_v$ chosen, the Taylor polynomials of $A$ and $v$ can be computed and $A_M$, $A_m$, $v_M$ and $v_m$ become known explicitly. The graphs of these bounding functions are given in Figure \ref{F4} and Figure \ref{F5}.
\begin{figure}[h]
	\begin{center}
		\includegraphics[width=10cm]{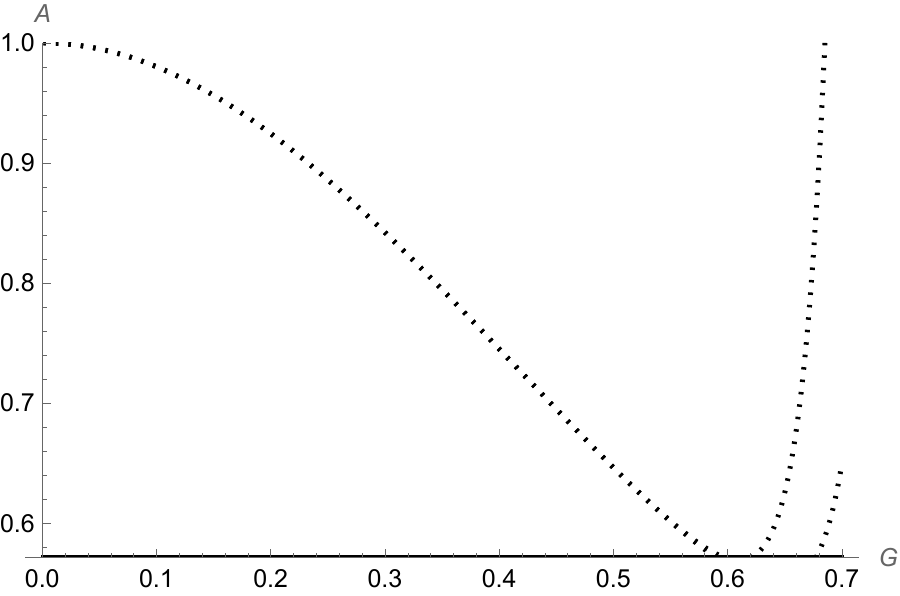}
	\end{center}
	\caption{This figure depicts $A_M(G)$ and $A_m(G)$ by the top and bottom dotted curves respectively. Note that these curves are almost indistinguishable until they cross the $A=\frac{4}{7}$ plane, which is given by the unbroken line at the bottom.}
	\label{F4}
\end{figure}
\begin{figure}[h]
	\begin{center}
		\includegraphics[width=10cm]{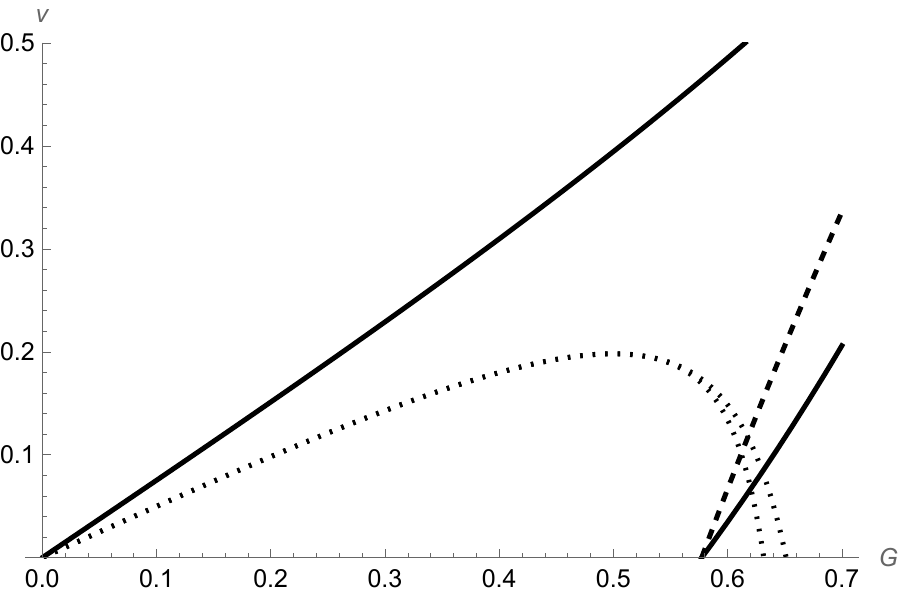}
	\end{center}
	\caption{This figure depicts $v_M(G)$ and $v_m(G)$ by the top and bottom dotted curves respectively. The Rankine--Hugoniot curve is given by the dashed curve and the singular and sonic surfaces are given as unbroken curves.}
	\label{F5}
\end{figure}

Even at 33rd order, Figure \ref{F5} shows that $v_M$ and $v_m$ noticeably diverge after passing the Rankine--Hugoniot curve. This is due to the trajectory approaching the sonic surface, where the solutions generically become singular, resulting in a slower convergence of the Taylor polynomials. With $A_M$ known explicitly, relation (\ref{5.12}) can be solved for $G_0$, at least approximately, to yield
\begin{align*}
	G_0 \approx 0.598 < \frac{3}{5}.
\end{align*}
However, we only need an upper bound on $G_0$, so we will use $\frac{3}{5}$ as the upper bound. This simple value has the advantage that we can compute $A_M(\frac{3}{5})$ exactly using symbolic manipulation software. This computation yields a negative rational number, analytically confirming $G_0<\frac{3}{5}$. We can also compute $v_m(\frac{3}{5})$, which then analytically confirms inequality (\ref{5.11}).

\subsubsection{Demonstrating the Perturbed Taylor Region is a Trapping Region}

The final step is to show that inequalities (\ref{5.13}) and (\ref{5.14}) hold on the interval $0<G<G_0$. To do this, the structure of equations (\ref{5.9}) and (\ref{5.10}) can be exploited, that is, it is possible to show:
\begin{align}
	\frac{\partial}{\partial v}\frac{dA}{dG} &< 0,\label{5.15}\\
	\frac{\partial}{\partial A}\frac{dv}{dG} &> 0,\label{5.16}
\end{align}
within the region bounded by (\ref{5.13}) and (\ref{5.14}). Starting with inequality (\ref{5.15}), we have
\begin{align*}
	\frac{\partial}{\partial v}\frac{dA}{dG} &= -\frac{4(1-A)v}{\{\cdot\}_S}\frac{\partial}{\partial v}\left(\xi\frac{dG}{d\xi}\right)^{-1} - \left(\xi\frac{dG}{d\xi}\right)^{-1}\frac{\partial}{\partial v}\frac{4(1-A)v}{\{\cdot\}_S}\\
	&= -\frac{4(1-A)G^2}{\{\cdot\}_S^3}\left(\xi\frac{dG}{d\xi}\right)^{-2}\left[4v^2(3-v^2)\left(\frac{1-A}{A}\right)+(3-v^2)\{\cdot\}_S-2(1-v^2)\{\cdot\}_S\left(\frac{1-A}{A}\right)\right]\\
	&< 0,
\end{align*}
which holds in the more general region described by $\frac{2}{5}<A<1$, $v>0$, $\{\cdot\}_S>0$ and $\{\cdot\}_D<0$. For inequality (\ref{5.16}), we have
\begin{align*}
	\frac{\partial}{\partial A}\frac{dv}{dG} &= -\left(\frac{1-v^2}{2\{\cdot\}_D}\right)\left[\{\cdot\}_S+4\left(\frac{1-A}{A}\right)\frac{\{\cdot\}_N}{\{\cdot\}_S}\right]\frac{\partial}{\partial A}\left(\xi\frac{dG}{d\xi}\right)^{-1}\\
	&- \left(\xi\frac{dG}{d\xi}\right)^{-1}\left(\frac{1-v^2}{2\{\cdot\}_D}\right)\frac{\partial}{\partial A}\left[\{\cdot\}_S+4\left(\frac{1-A}{A}\right)\frac{\{\cdot\}_N}{\{\cdot\}_S}\right]\\
	&= \frac{G}{A^2\{\cdot\}_S}\left(\xi\frac{dG}{d\xi}\right)^{-2}\left(\frac{1-v^2}{2\{\cdot\}_D}\right)\left[(2(1+v^2)G-4v)\{\cdot\}_S+4\{\cdot\}_N\right]\\
	&> 0,
\end{align*}
which holds in the region described by $v>0$, $\{\cdot\}_D<0$ and
\begin{align*}
	\left(2(1+v^2)G-4v\right)\{\cdot\}_S + 4\{\cdot\}_N < 0.
\end{align*}
This region is slightly smaller than the region described by $v>0$, $\{\cdot\}_D<0$ and $\{\cdot\}_S>0$, but includes the region bounded by (\ref{5.13}) and (\ref{5.14}) nonetheless. Now by Taylor's theorem, we know that (\ref{5.13}) and (\ref{5.14}) are satisfied on the interval $0<G<G_\epsilon$ for some small $G_\epsilon>0$. To demonstrate (\ref{5.13}) and (\ref{5.14}) on the interval $0<G<G_0$, it is thus sufficient to demonstrate:
\begin{align}
	\frac{d}{dG}(A_M-A)\Big\vert_{\substack{A=A_M\\ v=v_m}} &\geq 0,\label{5.17}\\
	\frac{d}{dG}(A-A_m)\Big\vert_{\substack{A=A_m\\ v=v_M}} &\geq 0,\label{5.18}\\
	\frac{d}{dG}(v_M-v)\Big\vert_{\substack{v=v_M\\ A=A_M}} &\geq 0,\label{5.19}\\
	\frac{d}{dG}(v-v_m)\Big\vert_{\substack{v=v_m\\ A=A_m}} &\geq 0,\label{5.20}
\end{align}
on the interval $G_\epsilon\leq G<G_0$. Note that the left hand sides of (\ref{5.17})--(\ref{5.20}) are functions of $A$, $v$ and $G$, so (\ref{5.15}) can be used to determine the most conservative value of $v$ in (\ref{5.17}) and (\ref{5.18}), and (\ref{5.16}) can be used to determine the most conservative value of $A$ in (\ref{5.19}) and (\ref{5.20}). In particular, the most conservative choice out of $v_M$ and $v_m$ for (\ref{5.17}) is $v_m$ and the most conservative choice for (\ref{5.18}) is $v_M$. Likewise, the the most conservative choice out of $A_M$ and $A_m$ for (\ref{5.19}) is $A_M$ and the most conservative choice for (\ref{5.20}) is $A_m$. This can be interpreted as remaining within the right wall of the trapping region implies remaining below the ceiling, remaining below the ceiling implies remaining within the left wall, remaining within the left wall implies remaining above the floor and remaining above the floor implies remaining within the right wall. Such an interpretation can be summarised as so:
\begin{align*}
	\begin{array}{ccc}
		A < A_M & \Rightarrow & v < v_M\\
		& &\\
		\Uparrow & & \Downarrow\\
		& &\\
		v > v_m & \Leftarrow & A > A_m
	\end{array}
\end{align*}
Thus if these implications are shown, a trajectory that begins within the trapping region remains within the trapping region. Now using the aforementioned conservative choices, the left hand sides of (\ref{5.17})--(\ref{5.20}) become explicitly known functions of $G$, which are depicted in Figures \ref{F6} and \ref{F7}. These figures also show where the functions become negative, if at all. A rigorous proof demonstrating the positivity of these functions on the interval $0<G<G_0$ is given in Appendix \ref{A.2} and involves the use of interval arithmetic. Therefore (\ref{5.13}) and (\ref{5.14}) hold on the interval $0<G<G_0$, and since the Lax characteristic conditions follow by Theorem \ref{T6}, the proof is complete.
\begin{figure}[h]
	\begin{center}
		\includegraphics[trim={0 0 0 0.5cm},clip,width=10cm]{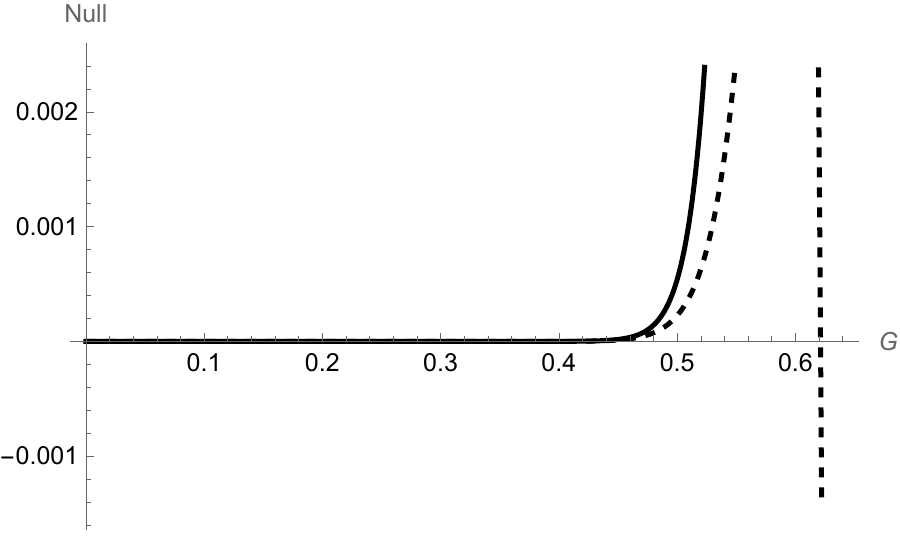}
	\end{center}
	\caption{This figure depicts $\frac{d}{dG}(A_M-A)\big\vert_{\substack{A=A_M\\ v=v_m}}$ and $\frac{d}{dG}(A-A_m)\big\vert_{\substack{A=A_m\\ v=v_M}}$ as unbroken and dashed curves respectively.}
	\label{F6}
\end{figure}
\begin{figure}[h]
	\begin{center}
		\includegraphics[trim={0 0 0 0.5cm},clip,width=10cm]{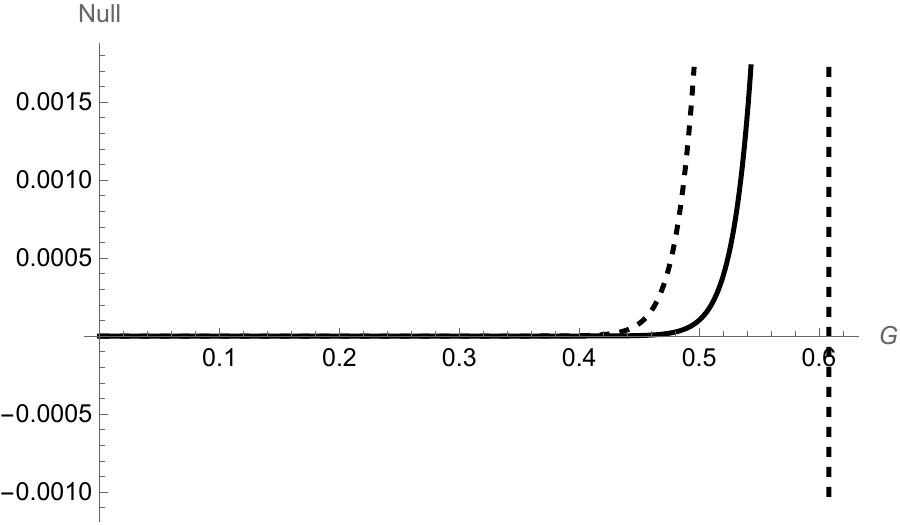}
	\end{center}
	\caption{This figure depicts $\frac{d}{dG}(v_M-v)\big\vert_{\substack{v=v_M\\ A=A_M}}$ and $\frac{d}{dG}(v-v_m)\big\vert_{\substack{v=v_m\\ A=A_m}}$ as unbroken and dashed curves respectively.}
	\label{F7}
\end{figure}

\section{Concluding Remarks}\label{Section6}

Throughout the preceding sections, the necessary machinery has been introduced, and in some places developed, to enable the construction of a new family of exact general relativistic shock waves that exhibit an accelerated expansion. This construction resolves the open problem of determining the expanding waves created behind a shock-wave explosion into a static isothermal sphere with an inverse square density profile, and in doing so, partially resolves a long-standing problem posed by Cahill and Taub. We saw in Section \ref{Section2} that this family of shock waves is one derivative less regular in Schwarzschild coordinates than it actually is, and that any delta function sources are cancelled within the Einstein tensor. Section \ref{Section3} built the machinery required for a phase-space analysis and determined the acceleration parameter of the pure radiation FLRW--TOV shock wave. Section \ref{Section4} extended the analysis to subluminal expanding waves and Section \ref{Section5} brought all this machinery together to establish the existence of a new family of general relativistic shock waves, with a rigorous proof given in the physically significant pure radiation case.

Given that it is possible to rigorously demonstrate the existence of a pure radiation shock wave, the obvious follow-up question is whether it is possible to rigorously demonstrate the existence of the full two-parameter family of FLRW$(\sigma,a)$--TOV$(\bar{\sigma})$ shock waves, with $\sigma$ included in this parameter count. For a certain range of values of $\sigma$ and $\bar{\sigma}$ there is every reason to expect this to be possible.
\begin{conjecture}
	For $0<\bar{\sigma}\leq\sigma\leq\frac{1}{3}$, there exists an $a>0$ such that FLRW$(\sigma,a)$ can be matched to TOV$(\bar{\sigma})$ to form a general relativistic shock wave.
\end{conjecture}
The resolution of this conjecture is one avenue of future research. The continuous dependence of the solution trajectories on the parameters means that existence is all but guaranteed for $\bar{\sigma},\sigma\approx\frac{1}{3}$ and $\bar{\sigma}\approx H(\sigma)$. Moreover, the existence proof in the pure radiation case is readily modified to demonstrate existence for any fixed pair $0<\bar{\sigma}\leq\sigma\leq\frac{1}{3}$. The difficulty arises when generalising the proof from fixed parameter values to two-dimensional parameter spaces, since conservative estimates need to be satisfied for all values of $\sigma$ and $\bar{\sigma}$ in such spaces. It is likely to be possible to construct such a proof by patching together many subproofs, demonstrating existence in small two-dimensional parameter spaces, although this process may be rather tedious.

Another avenue of future research is in regard to the possible cosmological implications of FLRW$(\sigma,a)$--TOV$(\bar{\sigma})$ shock waves. It is shown in Section \ref{Section3} that the pure radiation shock wave yields an acceleration parameter value of $a\approx2.58$. What this wave would evolve into at the end of the Radiation Dominated Epoch and what the resulting accelerated expansion would be in the Matter Dominated Epoch are open problems. The discussions of \cite{ST2004} and \cite{ST2012A} imply that FLRW$(\sigma,a)$--TOV$(\bar{\sigma})$ shock waves have shock fronts that lie within the Hubble radius, so these shock waves are not viable cosmological models in the Radiation Dominated Epoch since evidence of a vast primordial shock wave would already be apparent. As mentioned in the introduction, Smoller and Temple demonstrate in \cite{ST2004} that it is possible to construct a general relativistic shock wave, with a shock surface beyond the Hubble radius, by modelling the entire Universe as a finite mass explosion within the Schwarzschild radius of a time-reversed black hole. However, this is only completed for the explicitly known FLRW$(\sigma,1)$ spacetime on the interior and not with a pure radiation equation of state each side of the shock surface.

The possibility remains to construct a pure radiation general relativistic shock wave with a shock surface beyond the Hubble radius and determine the resulting rate of expansion. If after transition into the Matter Dominated Epoch the predicted rate of expansion lies within observed estimates, then this mechanism offers a mathematically independent derivation for the accelerated expansion observed today without a cosmological constant, and thus, without dark energy.

\appendix

\section{Appendix}

\subsection{Section 2 Proofs}\label{A.1}

\subsubsection*{Section 2.2}

\begin{proof}[Proof of Proposition \ref{P1}]
	Since it is known that the subfamily of static similarity solutions are unique, it is sufficient to demonstrate that solutions with zero Schwarzschild coordinate velocity are static. In this light, substituting $v\equiv0$ into equation (\ref{2.5}) implies $A\equiv A_0$ for some constant $A_0$. Furthermore, substituting $v\equiv0$ into equation (\ref{2.7}) requires
	\begin{align*}
		9\sigma - \frac{1}{4}(3+3\sigma)^2\left(\frac{1-A_0}{A_0}\right) = 0
	\end{align*}
	to ensure $v'\equiv0$. This means $A_0$ is given as:
	\begin{align*}
		A_0(\sigma) = 1 - 2M(\sigma),
	\end{align*}
	where
	\begin{align*}
		M(\sigma) = \frac{2\sigma}{1+6\sigma+\sigma^2}.
	\end{align*}
	Now substituting both $v\equiv 0$ and $A\equiv A_0$ into equation (\ref{2.6}) and solving for $G$ yields
	\begin{align*}
		G(\xi) = C_1\xi^\frac{1-\sigma}{1+\sigma}
	\end{align*}
	for some positive constant $C_1$. Putting these results together yields
	\begin{align*}
		ds^2 = -C_2\xi^{\frac{4\sigma}{1+\sigma}}dt^2 + \frac{1}{1-2M(\sigma)}dr^2 + r^2d\Omega^2
	\end{align*}
	for some positive constant $C_2$. The density is given by
	\begin{align*}
		\rho = \frac{2M(\sigma)}{\kappa r^2}.
	\end{align*}
	Note that because $v\equiv0$, the coordinate frame is comoving with the fluid. Finally, making the temporal transformation:
	\begin{align*}
		\tilde{t} &= \frac{1+\sigma}{1-\sigma}t^{\frac{1-\sigma}{1+\sigma}},\\
		\tilde{r} &= r,
	\end{align*}
	puts the metric in the explicitly static form
	\begin{align*}
		d\tilde{s}^2 = -C_2\tilde{r}^{\frac{4\sigma}{1+\sigma}}d\tilde{t}^2 + \frac{1}{1-2M(\sigma)}d\tilde{r}^2 + \tilde{r}^2d\Omega^2.
	\end{align*}
	Given that the density remains unchanged in its static form under this transformation, the proof is complete.
\end{proof}

\subsubsection*{Section 2.3}

\begin{proof}[Proof of Proposition \ref{P3}]
	To change FLRW$(\sigma,1)$ from self-similar comoving coordinates to self-similar Schwarzschild coordinates, the following coordinate transformation, given in \cite{CT1971}, is used:
	\begin{align}
		dt &= e^{-\mu}\left(e^\varphi\cosh\omega\, d\hat{t}+e^\psi\sinh\omega\, d\hat{r}\right),\label{2.13}\\
		dr &= e^{-\nu}\left(e^\varphi\sinh\omega\, d\hat{t}+e^\psi\cosh\omega\, d\hat{r}\right),\label{2.14}
	\end{align}
	where:
	\begin{align}
		\tanh\omega &= e^{\psi-\varphi}\frac{\partial_{\hat{t}}(\mathscr{R}\hat{r})}{\partial_{\hat{r}}(\mathscr{R}\hat{r})},\label{2.15}\\
		e^{-2\nu} &= e^{-2\psi}[\partial_{\hat{r}}(\mathscr{R}\hat{r})]^2 - e^{-2\varphi}[\partial_{\hat{t}}(\mathscr{R}\hat{r})]^2,\label{2.16}
	\end{align}
	and $\mu$ is such that $dt$ is a perfect differential. Relations (\ref{2.15}) and (\ref{2.16}) come from setting $r=\mathscr{R}\hat{r}$. The resulting self-similar Schwarzschild form of the metric is then given by
	\begin{align*}
		ds^2 = -e^{2\mu}dt^2 + e^{2\nu}dr^2 + r^2d\Omega^2.
	\end{align*}
	Starting with the $\tanh\omega$ term, we have
	\begin{align*}
		\tanh\omega &= e^{\psi-\varphi}\frac{\partial_{\hat{t}}(\mathscr{R}\hat{r})}{\partial_{\hat{r}}(\mathscr{R}\hat{r})}\\
		&= \beta^{-1}\gamma^{-1}\hat{\xi}^{-\frac{2}{3+3\sigma}}\frac{-\hat{\xi}^2\partial_{\hat{r}}\mathscr{R}}{\mathscr{R} + \hat{r}\partial_{\hat{r}}\mathscr{R}}\\
		&= \beta^{-1}\gamma^{-1}\hat{\xi}^{-\frac{2}{3+3\sigma}}\frac{-\hat{\xi}^2(-\frac{2}{3+3\sigma})\hat{\xi}^{-\frac{2}{3+3\sigma}-1}}{\hat{\xi}^{-\frac{2}{3+3\sigma}} + \hat{\xi}(-\frac{2}{3+3\sigma})\hat{\xi}^{-\frac{2}{3+3\sigma}-1}}\\
		&= 2(1+3\sigma)^{-1}\beta^{-1}\gamma^{-1}\hat{\xi}^{\frac{1+3\sigma}{3+3\sigma}},
	\end{align*}
	and this yields:
	\begin{align*}
		\cosh\omega &= (1-\tanh^2\omega)^{-\frac{1}{2}} = \left[1-4(1+3\sigma)^{-2}\beta^{-2}\gamma^{-2}\hat{\xi}^{\frac{2+6\sigma}{3+3\sigma}}\right]^{-\frac{1}{2}},\\
		\sinh\omega &= \tanh\omega(1-\tanh^2\omega)^{-\frac{1}{2}} = 2(1+3\sigma)^{-1}\beta^{-1}\gamma^{-1}\hat{\xi}^{\frac{1+3\sigma}{3+3\sigma}}\left[1-4(1+3\sigma)^{-2}\beta^{-2}\gamma^{-2}\hat{\xi}^{\frac{2+6\sigma}{3+3\sigma}}\right]^{-\frac{1}{2}}.
	\end{align*}
	The $e^{-2\nu}$ term is computed similarly, to obtain
	\begin{align*}
		e^{-2\nu} &= e^{-2\psi}[\partial_{\hat{r}}(\mathscr{R}\hat{r})]^2 - e^{-2\varphi}[\partial_{\hat{t}}(\mathscr{R}\hat{r})]^2\\
		&= \gamma^{2}\hat{\xi}^{\frac{4}{3+3\sigma}}\left[\hat{\xi}^{-\frac{2}{3+3\sigma}}+\hat{\xi}\left(-\frac{2}{3+3\sigma}\right)\hat{\xi}^{-\frac{2}{3+3\sigma}-1}\right]^2 - \beta^{-2}\left[-\hat{\xi}^2\left(-\frac{2}{3+3\sigma}\right)\hat{\xi}^{-\frac{2}{3+3\sigma}-1}\right]^2\\
		&= (1+3\sigma)^2(3+3\sigma)^{-2}\gamma^{2}\left[1-4(1+3\sigma)^{-2}\beta^{-2}\gamma^{-2}\hat{\xi}^{\frac{2+6\sigma}{3+3\sigma}}\right],
	\end{align*}
	and this yields:
	\begin{align*}
		dt &= \beta e^{-\mu}\left[1-4(1+3\sigma)^{-2}\beta^{-2}\gamma^{-2}\hat{\xi}^{\frac{2+6\sigma}{3+3\sigma}}\right]^{-\frac{1}{2}}\left[d\hat{t}+2(1+3\sigma)^{-1}\beta^{-2}\gamma^{-2}\hat{\xi}^{-1}\hat{\xi}^{\frac{2+6\sigma}{3+3\sigma}}d\hat{r}\right],\\
		dr &= 2(3+3\sigma)^{-1}\hat{\xi}^{\frac{1+3\sigma}{3+3\sigma}}\left[d\hat{t}+\frac{1}{2}(1+3\sigma)\hat{\xi}^{-1}d\hat{r}\right].
	\end{align*}
	Now given that $\mu$ is such that the right hand side of (\ref{2.13}) is a perfect differential, then
	\begin{align*}
		\frac{\partial}{\partial\hat{r}}e^{-\eta} = \frac{\partial}{\partial\hat{t}}\left[2(1+3\sigma)^{-1}\beta^{-2}\gamma^{-2}\hat{\xi}^{-1}\hat{\xi}^{\frac{2+6\sigma}{3+3\sigma}}e^{-\eta}\right],
	\end{align*}
	where
	\begin{align*}
		e^{-\eta} = e^{-\mu}\left[1-4(1+3\sigma)^{-2}\beta^{-2}\gamma^{-2}\hat{\xi}^{\frac{2+6\sigma}{3+3\sigma}}\right]^{-\frac{1}{2}}.
	\end{align*}
	The differential equation for $\eta$ is equivalent to
	\begin{align*}
		\frac{1}{\hat{t}}\frac{d}{d\hat{\xi}}e^{-\eta} = -\frac{\hat{r}}{\hat{t}^2}\frac{d}{d\hat{\xi}}\left[2(1+3\sigma)^{-1}\beta^{-2}\gamma^{-2}\hat{\xi}^{-1}\hat{\xi}^{\frac{2+6\sigma}{3+3\sigma}}e^{-\eta}\right],
	\end{align*}
	which implies
	\begin{align*}
		\eta' &= 2(1+3\sigma)^{-1}\beta^{-2}\gamma^{-2}\hat{\xi}e^\eta\left[-\hat{\xi}^{-1}\hat{\xi}^{\frac{2+6\sigma}{3+3\sigma}}\eta'e^{-\eta}+(3\sigma-1)(3+3\sigma)^{-1}\hat{\xi}^{-\frac{4}{3+3\sigma}}e^{-\eta}\right]\\
		&= 2(1+3\sigma)^{-1}(3\sigma-1)(3+3\sigma)^{-1}\beta^{-2}\gamma^{-2}\hat{\xi}^{-1}\hat{\xi}^{\frac{2+6\sigma}{3+3\sigma}}\left[1+2(1+3\sigma)^{-1}\beta^{-2}\gamma^{-2}\hat{\xi}^{\frac{2+6\sigma}{3+3\sigma}}\right]^{-1}.
	\end{align*}
	Integrating yields
	\begin{align*}
		\eta = (3\sigma-1)(2+6\sigma)^{-1}\log\left[1+2(1+3\sigma)^{-1}\beta^{-2}\gamma^{-2}\hat{\xi}^{\frac{2+6\sigma}{3+3\sigma}}\right] + C_3
	\end{align*}
	for some constant $C_3$, and thus
	\begin{align*}
		e^{-\eta} = \delta\left[1+2(1+3\sigma)^{-1}\beta^{-2}\gamma^{-2}\hat{\xi}^{\frac{2+6\sigma}{3+3\sigma}}\right]^{\frac{1-3\sigma}{2+6\sigma}}
	\end{align*}
	for some positive constant $\delta$. Plugging this into $dt$ then gives
	\begin{align*}
		dt = \delta\beta\left[1+2(1+3\sigma)^{-1}\beta^{-2}\gamma^{-2}\hat{\xi}^{\frac{2+6\sigma}{3+3\sigma}}\right]^{\frac{1-3\sigma}{2+6\sigma}}\left[d\hat{t}+2(1+3\sigma)^{-1}\beta^{-2}\gamma^{-2}\hat{\xi}^{-1}\hat{\xi}^{\frac{2+6\sigma}{3+3\sigma}}d\hat{r}\right].
	\end{align*}
	Because the transformation is taking the metric from one self-similar form to another, let
	\begin{align*}
		t = \mathscr{T}(\hat{\xi})\hat{t},
	\end{align*}
	so that:
	\begin{align*}
		\frac{\partial t}{\partial\hat{t}} &= \mathscr{T}(\hat{\xi}) - \hat{\xi}\mathscr{T}'(\hat{\xi}) = \delta\beta\left[1+2(1+3\sigma)^{-1}\beta^{-2}\gamma^{-2}\hat{\xi}^{\frac{2+6\sigma}{3+3\sigma}}\right]^{\frac{1-3\sigma}{2+6\sigma}},\\
		\frac{\partial t}{\partial\hat{r}} &= \mathscr{T}'(\hat{\xi}) = 2\delta(1+3\sigma)^{-1}\beta^{-1}\gamma^{-2}\hat{\xi}^{-1}\hat{\xi}^{\frac{2+6\sigma}{3+3\sigma}}\left[1+2(1+3\sigma)^{-1}\beta^{-2}\gamma^{-2}\hat{\xi}^{\frac{2+6\sigma}{3+3\sigma}}\right]^{\frac{1-3\sigma}{2+6\sigma}}.
	\end{align*}
	Solving these equations yields the same function for $\mathscr{T}(\hat{\xi})$ only when the integration constant is zero, therefore
	\begin{align*}
		\mathscr{T}(\hat{\xi}) = \delta\beta\left[1+2(1+3\sigma)^{-1}\beta^{-2}\gamma^{-2}\hat{\xi}^{\frac{2+6\sigma}{3+3\sigma}}\right]^{\frac{3+3\sigma}{2+6\sigma}}.
	\end{align*}
	Now the fluid four-velocity $\vec{u}$ is given in self-similar comoving coordinates as
	\begin{align*}
		\vec{u} = (\hat{u}^0,\hat{u}^1,\hat{u}^2,\hat{u}^3) = (e^{-\mu},0,0,0) = (\beta^{-1},0,0,0),
	\end{align*}
	and in self-similar Schwarzschild coordinates as
	\begin{align*}
		\vec{u} &= (u^0,u^1,u^2,u^3) = \left(\hat{u}^0\frac{\partial t}{\partial\hat{t}},\hat{u}^0\frac{\partial r}{\partial\hat{t}},0,0\right) = \left(\beta^{-1}\frac{\partial t}{\partial\hat{t}},\beta^{-1}\frac{\partial r}{\partial\hat{t}},0,0\right)\\
		&= \left(\delta\left[1+2(1+3\sigma)^{-1}\beta^{-2}\gamma^{-2}\hat{\xi}^{\frac{2+6\sigma}{3+3\sigma}}\right]^{\frac{1-3\sigma}{2+6\sigma}},2(3+3\sigma)^{-1}\beta^{-1}\hat{\xi}^{\frac{1+3\sigma}{3+3\sigma}},0,0\right).
	\end{align*}
	Therefore, by Definition \ref{D2}
	\begin{align*}
		v = e^{\nu-\mu}\frac{u^1}{u^0} = 2(3+3\sigma)^{-1}\beta^{-1}\hat{\xi}^{\frac{1+3\sigma}{3+3\sigma}}.
	\end{align*}
	Finally, by substituting in $\beta$ and $\gamma$ and noting
	\begin{align*}
		\xi = \frac{r}{t} = \frac{\mathscr{R}(\hat{\xi})\hat{r}}{\mathscr{T}(\hat{\xi})\hat{t}} = \frac{\mathscr{R}(\hat{\xi})}{\mathscr{T}(\hat{\xi})}\hat{\xi},
	\end{align*}
	the rest follows.
\end{proof}

\begin{proof}[Proof of Proposition \ref{P4}]
	First note (\ref{2.19}) is immediately obtained from Proposition \ref{P3}. Then by definition:
	\begin{align*}
		A &= e^{-2\nu} = 1 - \frac{2}{3}\hat{\xi}^{\frac{2+6\sigma}{3+3\sigma}},\\
		G &= \xi e^{\nu-\mu} = \frac{1}{\sqrt{6}}(3+3\sigma)\hat{\xi}^{\frac{1+3\sigma}{3+3\sigma}}\left[1+\frac{1}{3}(1+3\sigma)\hat{\xi}^{\frac{2+6\sigma}{3+3\sigma}}\right]^{-1},
	\end{align*}
	from which (\ref{2.17}) and (\ref{2.18}) follow. To check (\ref{2.17})--(\ref{2.19}) satisfy equations (\ref{2.5})--(\ref{2.7}), first show
	\begin{align*}
		\xi\frac{d}{d\xi} &= \xi\frac{d\hat{\xi}}{d\xi}\frac{d}{d\hat{\xi}} = \frac{(3+3\sigma)^2}{2+6\sigma}\frac{v}{AG}\hat{\xi}\frac{d}{d\hat{\xi}},
	\end{align*}
	and secondly show:
	\begin{align*}
		\hat{\xi}\frac{dA}{d\hat{\xi}} &= -\frac{2+6\sigma}{3+3\sigma}v^2,\\
		\hat{\xi}\frac{dG}{d\hat{\xi}} &= \frac{2+6\sigma}{(3+3\sigma)^2}\left(1-\frac{1}{2}(1+3\sigma)v^2\right)\frac{G^2}{v},\\
		\hat{\xi}\frac{dv}{d\hat{\xi}} &= \frac{1+3\sigma}{3+3\sigma}v.
	\end{align*}
	Then by Proposition \ref{P4} it is not difficult to confirm (\ref{2.17})--(\ref{2.19}) solve equations (\ref{2.5})--(\ref{2.7}).
\end{proof}

\subsubsection*{Section 2.5}

\begin{proof}[Proof of Proposition \ref{P6}]
	If each component of the stress-energy-momentum tensor $T$ is differentiable in $U$, then the conservation of mass-energy and momentum in $U$ is given by
	\begin{align*}
		\nabla_\nu T^{\mu\nu} = 0.
	\end{align*}
	These conditions are equivalent to
	\begin{align*}
		\int_U \varphi\nabla_\nu T^{\mu\nu}\, d\vec{x} = 0\quad \forall\ \varphi \in C^\infty_c(U),
	\end{align*}
	and by using the identity
	\begin{align*}
		\varphi\nabla_\nu T^{\mu\nu} = \nabla_\nu(\varphi T^{\mu\nu}) - T^{\mu\nu}\nabla_\nu\varphi,
	\end{align*}
	are then equivalent to
	\begin{align*}
		\int_U \nabla_\nu(\varphi T^{\mu\nu})\, d\vec{x} - \int_U T^{\mu\nu}\nabla_\nu\varphi\, d\vec{x} = 0\quad \forall\ \varphi \in C^\infty_c(U).
	\end{align*}
	Now since $\varphi$ is compactly supported within $U$, the divergence theorem implies
	\begin{align*}
		\int_U \nabla_\nu(\varphi T^{\mu\nu})\, d\vec{x} = \int_{\partial U} \varphi T^{\mu\nu}n_\nu\, d\boldsymbol{x} = 0\quad \forall\ \varphi \in C^\infty_c(U),
	\end{align*}
	where $\vec{n}$ denotes the outward normal vector and $\boldsymbol{x}$ denotes the restriction of the $\vec{x}$ coordinates to $\Sigma$. Thus (\ref{2.34}) yields the weak form of the conservation of mass-energy and momentum across $\Sigma\cap U$. Conditions (\ref{2.35}) then follow from equation (\ref{2.1}).
\end{proof}

\begin{proof}[Proof of Proposition \ref{P7}]
	Let $U=U_1\cup U_2$ where $\partial U_1\cap\partial U_2=\Sigma\cap U$ and assume $g$ and $\bar{g}$ are sufficiently regular on their respective side of $\Sigma$, then
	\begin{align*}
		\int_{U} G^{\mu\nu}\nabla_\nu\varphi\, d\vec{x} &= \int_{U_1} G^{\mu\nu}(g)\nabla_\nu\varphi\, d\vec{x} + \int_{U_2} G^{\mu\nu}(\bar{g})\bar{\nabla}_\nu\varphi\, d\vec{x}\\
		&= \int_{U_1} \nabla_\nu\big(\varphi G^{\mu\nu}(g)\big)\, d\vec{x} - \int_{U_1} \varphi\nabla_\nu G^{\mu\nu}(g)\, d\vec{x}\\
		&+ \int_{U_2} \bar{\nabla}_\nu\big(\varphi G^{\mu\nu}(\bar{g})\big)\, d\vec{x} - \int_{U_2} \varphi\bar{\nabla}_\nu G^{\mu\nu}(\bar{g})\, d\vec{x}\\
		&= \int_{\partial U_1} \varphi G^{\mu\nu}(g)n_\nu\, d\boldsymbol{x} - \int_{\partial U_2} \varphi G^{\mu\nu}(\bar{g})\bar{n}_\nu\, d\boldsymbol{x}\\
		&= \int_\Sigma \varphi G^{\mu\nu}(g)n_\nu\, d\boldsymbol{x} - \int_\Sigma \varphi G^{\mu\nu}(\bar{g})\bar{n}_\nu\, d\boldsymbol{x}\\
		&= \int_\Sigma \varphi[G^{\mu\nu}]n_\nu\, d\boldsymbol{x}\quad \forall\ \varphi \in C^\infty_c(U),
	\end{align*}
	where $\vec{n}$ denotes the outward normal vector and $\boldsymbol{x}$ denotes the restriction of the $\vec{x}$ coordinates to $\Sigma$. Thus
	\begin{align*}
		[G^{\mu\nu}]n_\nu = 0
	\end{align*}
	if and only if
	\begin{align*}
		\int_{U} G^{\mu\nu}\nabla_\nu\varphi\, d\vec{x} = 0\quad \forall\ \varphi \in C^\infty_c(U).
	\end{align*}
\end{proof}

\subsubsection*{Section 2.7}

\begin{proof}[Proof of Lemma \ref{L1}]
	This proof largely follows an analogous proof provided in \cite{ST1995}. To begin, recall that the speed of a shock is a coordinate dependent quantity that can be interpreted in a special relativistic sense at a point $\vec{p}$ in coordinate systems for which
	\begin{align}
		d\tilde{s}^2 = -d\tilde{t}^2 + d\tilde{r}^2 + \tilde{r}_0^2d\Omega^2,\label{2.37}
	\end{align}
	where $\tilde{r}_0$ is the value of $\tilde{r}$ at $\vec{p}$. In a locally Minkowskian coordinate frame, a speed at $\vec{p}$ transforms according to the special relativistic velocity transformation law when a Lorentz transformation is performed. The shock speed at a point $\vec{p}$ on the shock in a locally Minkowskian frame that is comoving with the interior fluid will now be determined. To this end, let $\hat{r}=\Phi(\hat{t})$ be the position of the shock in $(\hat{t},\hat{r})$ coordinates and let $(\tilde{t},\tilde{r})$ coordinates correspond to a locally Minkowskian system at $\vec{p}$ obtained from $(\hat{t},\hat{r})$ by a transformation of the form:
	\begin{align*}
		\tilde{t} &= \tilde{t}(\hat{t}),\\
		\tilde{r} &= \tilde{r}(\hat{r}),
	\end{align*}
	so that, in $(\tilde{t},\tilde{r})$ coordinates
	\begin{align*}
		d\tilde{s}^2 = -e^{2\varphi}\left(\frac{d\hat{t}}{d\tilde{t}}\right)^2d\tilde{t}^2 + e^{2\psi}\left(\frac{d\hat{r}}{d\tilde{r}}\right)^2d\tilde{r}^2 + \mathscr{R}^2\hat{r}^2d\Omega^2.
	\end{align*}
	Choose $(\tilde{t},\tilde{r})$ so that:
	\begin{align*}
		\frac{d\tilde{t}}{d\hat{t}} &= e^\varphi, & \frac{d\tilde{r}}{d\hat{r}} &= e^\psi.
	\end{align*}
	Then in $(\tilde{t},\tilde{r})$ coordinates at $\vec{p}$ the metric takes the form of (\ref{2.37}). The $(\tilde{t},\tilde{r})$ coordinates represent the class of locally Minkowskian coordinate frames that are fixed relative to the fluid particles of the interior spacetime at the point $\vec{p}$, that is, any two members of this class of coordinate frames differ only by higher order terms that do not affect the calculation of radial velocities at $\vec{p}$. Thus the speed $\dot{\tilde{r}}$ of a particle in $(\tilde{t},\tilde{r})$ coordinates gives the value of the speed of the particle relative to the interior fluid in the special relativistic sense. If the speed of a particle in $(\hat{t},\hat{r})$ coordinates is $\dot{\hat{r}}$, then its geometric speed relative to observers fixed with the interior fluid, and hence also fixed relative to the radial coordinate $\hat{r}$ of the metric $g$ because the fluid is comoving, is equal to
	\begin{align*}
		e^{\psi-\varphi}\dot{\hat{r}},
	\end{align*}
	since
	\begin{align}
		\frac{d\hat{r}}{d\hat{t}} = \frac{d\hat{r}}{d\tilde{r}}\frac{d\tilde{t}}{d\hat{t}}\frac{d\tilde{r}}{d\tilde{t}} = e^{\varphi-\psi}\frac{d\tilde{r}}{d\tilde{t}}.\label{2.38}
	\end{align}
	Now considering the shock wave moves with speed $\dot{\Phi}$, therefore by (\ref{2.38}) the speed of the shock relative to the interior fluid particles must be given by (\ref{2.36}), which completes the proof.
\end{proof}

\begin{proof}[Proof of Proposition \ref{P8}]
	This proof largely follows an analogous proof provided in \cite{ST1995}. Since the shock wave is expanding, the Lax characteristic conditions are given by (\ref{2.39}), and by Lemma \ref{L1}, the shock speed is given by (\ref{2.36}). As we are working in $(\tilde{t},\tilde{r})$ coordinates, $\tilde{\lambda}_{Int}^+$ is already known, so it remains to determine $\tilde{\lambda}_{Ext}^+$. Let $(\hat{v}^0,\hat{v}^1)$, $(\bar{v}^0,\bar{v}^1)$ and $(\tilde{v}^0,\tilde{v}^1)$ denote the non-zero components of the exterior fluid four-velocity given in interior comoving, exterior comoving and interior locally Minkowskian coordinates respectively. Since the aim is to compute the characteristic speed, which is a ratio of two vector components, a tangent vector of any length is sufficient. By writing $(\hat{x}^0,\hat{x}^1)=(\hat{t},\hat{r})$ and $(\bar{x}^0,\bar{x}^1)=(\bar{t},\bar{r})$, then
	\begin{align*}
		\hat{v}^\mu = \frac{\partial\hat{x}^\mu}{\partial\bar{x}^\nu}\bar{v}^\nu = \frac{\partial\hat{x}^\mu}{\partial\bar{x}^0}\bar{v}^0 = \frac{\partial\hat{x}^\mu}{\partial\bar{x}^0}.
	\end{align*}
	In light of this, the speed of the exterior fluid as measured in the interior coordinates $(\hat{t},\hat{r})$ is given by
	\begin{align*}
		\hat{w} = \frac{\hat{v}^1}{\hat{v}^0} = \frac{\partial\hat{x}^1}{\partial\bar{x}^0}\left(\frac{\partial\hat{x}^0}{\partial\bar{x}^0}\right)^{-1} = \frac{\partial\hat{r}}{\partial\bar{t}}\left(\frac{\partial\hat{t}}{\partial\bar{t}}\right)^{-1},
	\end{align*}
	and so, by (\ref{2.38}),
	\begin{align*}
		\tilde{w}=e^{\psi-\varphi}\hat{w}.
	\end{align*}
	This gives the exterior fluid speed in $(\tilde{t},\tilde{r})$ coordinates, and since the sound speed in the exterior spacetime is given by
	\begin{align*}
		\sqrt{\frac{d\bar{p}}{d\bar{\rho}}},
	\end{align*}
	the relativistic addition of velocities formula yields
	\begin{align*}
		\tilde{\lambda}^+_{Ext} = \frac{\tilde{w}+\sqrt{\frac{d\bar{p}}{d\bar{\rho}}}}{1+\tilde{w}\sqrt{\frac{d\bar{p}}{d\bar{\rho}}}},
	\end{align*}
	which completes the proof.
\end{proof}

\subsection{Interval Arithmetic}\label{A.2}

\subsubsection*{Inequality \ref{5.17}}

We begin by showing inequality (\ref{5.17}), that is,
\begin{align*}
	\frac{d}{dG}(A_M-A)\Big\vert_{\substack{A=A_M\\ v=v_m}} \geq 0
\end{align*}
on the interval $0<G<\frac{3}{5}$ for $\sigma=\frac{1}{3}$ and $a=\frac{17}{6}$. The expression $(A_M-A)'$ with $A=A_M$ and $v=v_m$, the conservative choice of $v$, is an odd 131st order polynomial in $G$ with a leading order of 33. We thus define
\begin{align*}
	p_1(x) = p_1\left(\frac{25}{9}G^2\right) = G^{-33}(A_M-A)'(G),
\end{align*}
so that $p_1(x)$ is an order 49 polynomial with leading order 0. Writing the coefficients to one decimal place, we find
\begin{align*}
	p_1(x) &= 67161.1 + 1208390.0x + 1393431.2x^2 - 4372737.8x^3 + 5986240.9x^4\\
	&- 6752795.3x^5 + 7064632.4x^6 - 7007877.7x^7 + 6827093.5x^8 - 6488710.6x^9\\
	&+ 6137919.2x^{10} - 5725506.5x^{11} + 5336513.1x^{12} - 4930984.2x^{13} + 4558174.4x^{14}\\
	&- 4189622.5x^{15} + 3860994.3x^{16} - 3135788.5x^{17} + 2244699.6x^{18} - 1574639.5x^{19}\\
	&+ 1160577.0x^{20} - 854498.4x^{21} + 644314.6x^{22} - 479718.5x^{23} + 361405.8x^{24}\\
	&- 266658.0x^{25} + 196805.5x^{26} - 140476.7x^{27} + 98421.5x^{28} - 64541.8x^{29}\\
	&+ 39206.0x^{30} - 18931.4x^{31} + 3928.6x^{32} + 8213.2x^{33} - 8265.4x^{34}\\
	&+ 5774.5x^{35} - 4029.6x^{36} + 3075.0x^{37} - 2250.7x^{38} + 1747.4x^{39}\\
	&- 1300.2x^{40} + 1005.9x^{41} - 746.5x^{42} + 567.0x^{43} - 411.1x^{44}\\
	&+ 300.0x^{45} - 204.7x^{46} + 136.1x^{47} - 77.8x^{48} + 36.0x^{49}.
\end{align*}
Given that we wish to show $p_1(x)\geq0$ for $0<x<1$, we can make conservative adjustments to this polynomial by cancelling positive lower order coefficients with negative higher order coefficients. In addition, we round down the positive coefficients and round up the negative coefficients. Thus, for $0<x<1$, we notice that the positivity of $\bar{p}_1(x)$ implies the positivity of $p_1(x)$, where
\begin{align*}
	\bar{p}_1(x) &= 67161 + 1208389x + 1393432x^2 - 4372738x^3 + 5986240x^4\\
	&- 6752796x^5 + 7064632x^6 - 7007878x^7 + 6827093x^8 - 6488711x^9\\
	&+ 6137919x^{10} - 5725507x^{11} + 5336513x^{12} - 4930985x^{13} + 4558174x^{14}\\
	&- 4189623x^{15} + 3860994x^{16} - 3135789x^{17} + 2244699x^{18} - 1574640x^{19}.
\end{align*}
Taking the derivative, we obtain
\begin{align*}
	\bar{p}_1'(x) &= 1208389 + 2786862x - 13118214x^2 + 23944960x^3\\
	&- 33763980x^4 + 42387792x^5 - 49055146x^6 + 54616744x^7 - 58398399x^8\\
	&+ 61379190x^9 - 62980577x^{10} + 64038156x^{11} - 64102805x^{12} + 63814436x^{13}\\
	&- 62844345x^{14} + 61775904x^{15} - 53308413x^{16} + 40404582x^{17} - 29918160x^{18},
\end{align*}
and we can use this to obtain a crude analytic lower bound on the derivative by grouping the terms, hence
\begin{align*}
	\min_{0\leq x\leq1}\bar{p}_1'(x) &\geq (1208389+2786862-13118214) + (23944960-33763980)\\
	&+ (42387792-49055146) + (54616744-58398399)\\
	&+ (61379190-62980577) + (64038156-64102805)\\
	&= -31057028.
\end{align*}
This represents the maximum possible decrease of $\bar{p}_1(x)$ on the interval $0<x<1$. Now given that $\bar{p}_1(0)=67161$, providing $\bar{p}_1(x_n)\geq67161$ for each point $x_n$ separated by at most
\begin{align*}
	\frac{67161}{31057028} > \frac{1}{500} = 0.002,
\end{align*}
then we can guarantee the positivity of $\bar{p}_1(x)$ on $0<x<1$. This is checked using Wolfram Mathematica Version 14.1 with the input:
\begin{verbatim*}
P1[x_]:=67161+1208389x+1393431x^2-4372738x^3+5986240x^4
-6752796x^5+7064632x^6-7007878x^7+6827093x^8-6488711x^9
+6137919x^10-5725507x^11+5336513x^12-4930985x^13+4558174x^14
-4189623x^15+3860994x^16-3135789x^17+2244699x^18-1574640x^19

For[x=0,P1[x]>67160&&x<1,x=x+0.002]

Print[x]
\end{verbatim*}
The output is 1, confirming that $\bar{p}_1(x_n)\geq67161$ with $x_n=0.002n$ for each $n\in\{0,1,\dots,499,500\}$. Therefore inequality (\ref{5.17}) holds on the interval $0<G<\frac{3}{5}$.

\subsubsection*{Inequality \ref{5.18}}

We now turn to inequality (\ref{5.17}), that is,
\begin{align*}
	\frac{d}{dG}(A-A_m)\Big\vert_{\substack{A=A_m\\ v=v_M}} \geq 0
\end{align*}
on the interval $0<G<\frac{3}{5}$ for $\sigma=\frac{1}{3}$ and $a=\frac{17}{6}$. The expression $(A-A_m)'$ with $A=A_m$ and $v=v_M$, the new conservative choice of $v$, is similarly an odd 131st order polynomial in $G$ with a leading order of 33. We thus define
\begin{align*}
	p_2(x) = p_2\left(\frac{25}{9}G^2\right) = G^{-33}(A-A_m)'(G),
\end{align*}
so that $p_2(x)$ is also an order 49 polynomial with leading order 0. Again, writing the coefficients to one decimal place unless the coefficient is less than 0.1, we find
\begin{align*}
	p_2(x) &= 2149156.3 - 4468591.8x^1 + 4222794.8x^2 - 3293456.4x^3 + 2598711.7x^4\\
	&- 2053219.8x^5 + 1630902.2x^6 - 1302193.6x^7 + 1038967.6x^8 - 832826.9x^9\\
	&+ 664538.5x^{10} - 531758.6x^{11} + 422419.4x^{12} - 335521.6x^{13} + 263840.6x^{14}\\
	&- 206528.6x^{15} + 143457.3x^{16} - 86328.7x^{17} + 54454.2x^{18} - 37413.4x^{19}\\
	&+ 25471.0x^{20} - 17559.1x^{21} + 11842.8x^{22} - 7987.2x^{23} + 5203.8x^{24}\\
	&- 3338.6x^{25} + 2020.7x^{26} - 1164.5x^{27} + 589.0x^{28} - 243.8x^{29}\\
	&+ 38.7x^{30} + 55.4x^{31} - 43.1x^{32} + 1.7x^{33} + 2.7x^{34}\\
	&- 2.2x^{35} + 1.1x^{36} - 0.9x^{37} + 0.4x^{38} - 0.3x^{39}\\
	&+ 0.1x^{40} - 0.095x^{41} + 0.027x^{42} - 0.016x^{43} - 0.0046x^{44}\\
	&+ 0.0039x^{45} - 0.0075x^{46} + 0.0040x^{47} - 0.0027x^{48} - 0.0021x^{49}.
\end{align*}
We again make conservative adjustments to this polynomial by cancelling positive lower order coefficients with negative higher order coefficients, rounding down the positive coefficients and rounding up the negative coefficients. Thus, for $0<x<1$, the positivity of $\bar{p}_2(x)$ implies the positivity of $p_2(x)$, where
\begin{align*}
	\bar{p}_2(x) &= 2149156 - 4468592x + 4222794x^2 - 3293457x^3 + 2598711x^4\\
	&- 2053220x^5 + 1630902x^6 - 1302194x^7 + 1038967x^8 - 832827x^9\\
	&+ 664538x^{10} - 531759x^{11} + 422419x^{12} - 335522x^{13} + 263840x^{14}\\
	&- 206529x^{15} + 143457x^{16} - 86329x^{17}.
\end{align*}
Taking the derivative, we obtain
\begin{align*}
	\bar{p}_2'(x) &= -4468592 + 8445588x - 9880371x^2 + 10394844x^3 - 10266100x^4\\
	&+ 9785412x^5 - 9115358x^6 + 8311736x^7 - 7495443x^8 + 6645380x^9\\
	&- 5849349x^{10} + 5069028x^{11} - 4361786x^{12} + 3693760x^{13} - 3097935x^{14}\\
	&+ 2295312x^{15} - 1467593x^{16},
\end{align*}
and again we can use this to obtain a crude analytic lower bound on the derivative by grouping the terms, hence
\begin{align*}
	\min_{0\leq x\leq1}\bar{p}_2'(x) \geq -4468592 + (8445588-9880371) = -5903375.
\end{align*}
As like before, this represents the maximum possible decrease of $\bar{p}_2(x)$ on the interval $0<x<1$, and given that $\bar{p}_2(1)=24355$, providing $\bar{p}_2(x_n)\geq24355$ for each point $x_n$ separated by at most
\begin{align*}
	\frac{24355}{5903375} > \frac{1}{250} = 0.004,
\end{align*}
then we can guarantee the positivity of $\bar{p}_2(x)$ on $0<x<1$. This is again checked using Wolfram Mathematica Version 14.1 with the input:
\begin{verbatim*}
P2[x_]:=2149156-4468592x+4222794x^2-3293457x^3+2598711x^4
-2053220x^5+1630902x^6-1302194x^7+1038967x^8-832827x^9
+664538x^10-531759x^11+422419x^12-335522x^13+263840x^14
-206529x^15+143457x^16-86329x^17

For[x=0,P2[x]>24354&&x<1,x=x+0.004]

Print[x]
\end{verbatim*}
The output is again 1, confirming that $\bar{p}_2(x_n)\geq24355$ with $x_n=0.004n$ for each $n\in\{0,1,\dots,249,250\}$. Therefore inequality (\ref{5.18}) holds on the interval $0<G<\frac{3}{5}$.

\subsubsection*{Inequality \ref{5.19}}

We now consider inequality \ref{5.19}, given by
\begin{align*}
	\frac{d}{dG}(v_M-v)\Big\vert_{\substack{v=v_M\\ A=A_M}} \geq 0
\end{align*}
on the interval $0<G<\frac{3}{5}$ for $\sigma=\frac{1}{3}$ and $a=\frac{17}{6}$. The expression $(v_M-v)'$ with $v=v_M$ and $A=A_m$, the conservative choice of $A$, is an even 232nd order polynomial in $G$ with a leading order of 34. We thus define
\begin{align*}
	p_3(x) = p_3\left(\frac{25}{9}G^2\right) = G^{-34}(v_M-v)'(G),
\end{align*}
so that $p_3(x)$ is an order 99 polynomial with leading order 0. Writing the coefficients to one decimal place, we find
\begin{align*}
	p_3(x) &= 484998.3 + 361395.0x - 66569.6x^2 + 166861.7x^3 - 1299.9x^4\\
	&+ 100239.1x^5 - 3234.9x^6 + 64602.3x^7 - 2871.0x^8 + 42298.1x^9\\
	&- 2142.1x^{10} + 27696.3x^{11} - 1601.2x^{12} + 17912.5x^{13} - 1275.0x^{14}\\
	&+ 11297.8x^{15} - 3897.4x^{16} + 3735.6x^{17} + 369.1x^{18} + 2142.1x^{19}\\
	&+ 218.7x^{20} + 1214.4x^{21} + 229.8x^{22} + 652.8x^{23} + 181.5x^{24}\\
	&+ 335.3x^{25} + 127.0x^{26} + 163.0x^{27} + 83.3x^{28} + 74.5x^{29}\\
	&+ 52.5x^{30} + 33.0x^{31} + 38.0x^{32} + 27.3x^{33} + 18.0x^{34}\\
	&+ 13.8x^{35} + 9.6x^{36} + 7.0x^{37} + 4.9x^{38} + 3.5x^{39}\\
	&+ 2.4x^{40} + 1.7x^{41} + 1.2x^{42} + 0.9x^{43} + 0.6x^{44}\\
	&+ 0.4x^{45} + 0.3x^{46} + 0.2x^{47} + 0.1x^{48} + O\left(x^{49}\right),
\end{align*}
with all higher order coefficients less than $0.1$. We immediately see
\begin{align*}
	p_3(x) &\geq 484998.3 + (361395.0x-66569.6x^2) + (166861.7x^3-1299.9x^4)\\
	&+ (100239.1x^5-3234.9x^6) + (64602.3x^7-2871.0x^8) + (42298.1x^9-2142.1x^{10})\\
	&+ (27696.3x^{11}-1601.2x^{12}) + (17912.5x^{13}-1275.0x^{14}) + (11297.8x^{15}-3897.4x^{16}) - 5x^{49}\\
	&\geq 0
\end{align*}
on the interval $0<x<1$, and therefore inequality (\ref{5.19}) holds on the interval $0<G<\frac{3}{5}$.

\subsubsection*{Inequality \ref{5.20}}

Inequality \ref{5.20}, given by
\begin{align*}
	\frac{d}{dG}(v-v_m)\Big\vert_{\substack{v=v_m\\ A=A_m}} \geq 0
\end{align*}
on the interval $0<G<\frac{3}{5}$ for $\sigma=\frac{1}{3}$ and $a=\frac{17}{6}$ will follow similarly to inequality (\ref{5.19}). Again, the expression $(v-v_m)'$ with $v=v_m$ and $A=A_M$, the conservative choice of $A$, is an even 232nd order polynomial in $G$ with a leading order of 34. We thus define
\begin{align*}
	p_4(x) = p_4\left(\frac{25}{9}G^2\right) = G^{-34}(v-v_m)'(G),
\end{align*}
so that $p_4(x)$ is an order 99 polynomial with leading order 0. Again, writing the coefficients to one decimal place, we find
\begin{align*}
	p_4(x) &= 21191189.9 - 20457561.7x + 2548849.3x^2 + 1846153.2x^3 - 2441295.7x^4\\
	&+ 1556401.3x^5 - 1595278.6x^6 + 1021111.5x^7 - 1069470.5x^8 + 689631.7x^9\\
	&- 747916.5x^{10} + 483306.2x^{11} - 539631.7x^{12} + 349697.9x^{13} - 398158.9x^{14}\\
	&+ 259603.2x^{15} + 1747036.7x^{16} - 4163811.0x^{17} + 3679189.6x^{18} - 2513361.1x^{19}\\
	&+ 1596962.7x^{20} - 1153474.2x^{21} + 797367.3x^{22} - 637410.9x^{23} + 459013.5x^{24}\\
	&- 389728.9x^{25} + 287585.4x^{26} - 253413.4x^{27} + 190620.5x^{28} - 171795.7x^{29}\\
	&+ 131448.7x^{30} - 120084.3x^{31} + 94009.2x^{32} - 71266.9x^{33} + 8328.2x^{34}\\
	&+ 8150.7x^{35} - 12764.2x^{36} + 4258.1x^{37} - 6820.7x^{38} + 2054.6x^{39}\\
	&- 4028.7x^{40} + 1224.5x^{41} - 2512.4x^{42} + 834.4x^{43} - 1632.8x^{44}\\
	&+ 607.8x^{45} - 1099.8x^{46} + 455.9x^{47} - 764.7x^{48} - 889.8x^{49}\\
	&+ 1305.2x^{50} - 1470.1x^{51} + O\left(z^{52}\right),
\end{align*}
with all higher order coefficients less than $1000$. We can group these terms in a conservative manner to yield
\begin{align*}
	p_4(x) &\geq 700000 + (20491189.9-20457561.7x)\\
	&+ 100000x^2 + (2448849.3x^2-2441295.7x^4)\\
	&+ 250000x^3 + (1596153.2x^3-1595278.6x^6)\\
	&+ 400000x^5 + (1156401.3x^5-1069470.5x^8)\\
	&+ 250000x^7 + (771111.5x^7-747916.5x^{10})\\
	&+ 150000x^9 + (539631.7x^9-539631.7x^{12})\\
	&+ 50000x^{11} + (433306.2x^{11}-398158.9x^{14})\\
	&+ 349697.9x^{13} + 259603.2x^{15} + 1747036.7x^{16} - 4163811.0x^{17}\\
	&+ (3679189.6x^{18}-2513361.1x^{19}) + (1596962.7x^{20}-1153474.2x^{21})\\
	&+ (797367.3x^{22}-637410.9x^{23}) + (459013.5x^{24}-389728.9x^{25})\\
	&+ (287585.4x^{26}-253413.4x^{27}) + (190620.5x^{28}-171795.7x^{29})\\
	&+ (131448.7x^{30}-120084.3x^{31}) + (94009.2x^{32}-71266.9x^{33}-12764.2x^{36})\\
	&+ (8328.2x^{34}-6820.7x^{38}) + (8150.7x^{35}-4028.7x^{40}) + (4258.1x^{37}-2512.4x^{42})\\
	&+ (2054.6x^{39}-1632.8x^{44}) + (1224.5x^{41}-1099.8x^{46}) + (834.4x^{43}-764.7x^{48})\\
	&+ (607.8x^{45}+455.9x^{47}-889.8x^{49}+1305.2x^{50}-1470.1x^{51}) - 50000x^{52}\\
	&\geq 0,
\end{align*}
noting that the ungrouped terms form a strictly positive polynomial on $0<x<1$. Therefore inequality (\ref{5.20}) holds on the interval $0<G<\frac{3}{5}$.

\end{document}